\documentclass[11pt,letterpaper,dvipsnames]{article}
\usepackage[margin=25mm]{geometry}
\usepackage{amsmath}
\usepackage{amsfonts}
\usepackage{mathtools}
\usepackage{amsthm}
\usepackage{amssymb}
\usepackage{algorithm}
\usepackage{algorithmicx}
\usepackage{algpseudocode}
\usepackage{enumitem}
\usepackage{hyperref}
\usepackage[capitalise,nameinlink,noabbrev]{cleveref}
\usepackage{bm}
\usepackage{bbm}
\usepackage{tikz}
\usepackage{thmtools} 
\usepackage{thm-restate}
\usepackage{tocloft}
\usepackage[nottoc]{tocbibind} 
\usepackage{xspace}

\hypersetup{
	colorlinks=true,
	urlcolor=MidnightBlue,
	linkcolor=MidnightBlue,
	citecolor=MidnightBlue,
	unicode
}

\newcounter{sharedcounter}
\newtheorem{theorem}[sharedcounter]{Theorem}
\newtheorem{lemma}[sharedcounter]{Lemma}
\newtheorem{claim}[sharedcounter]{Claim}
\newtheorem{definition}[sharedcounter]{Definition}
\newtheorem{corollary}[sharedcounter]{Corollary}
\newtheorem{observation}{Observation}
\newtheorem{question}{Question}

\crefname{equation}{}{}
\crefname{line}{line}{lines}
\crefname{claim}{Claim}{Claims}

\algnewcommand\algorithmicinput{\textbf{Input: }}
\algnewcommand\Input{\item[\algorithmicinput]}
\algnewcommand\algorithmicoutput{\textbf{Output: }}
\algnewcommand\Output{\item[\algorithmicoutput]}

\let\OLDthebibliography\thebibliography 
\renewcommand\thebibliography[1]{
  \OLDthebibliography{#1}
  \setlength{\itemsep}{.4ex}
}

\newcommand{\newmeasure}[2]{\newcommand{#1}{{\textup{\sffamily #2}}\xspace}}
\newmeasure{\D}{D}
\newmeasure{\R}{R}
\newmeasure{\C}{C}
\newcommand{\DDT}{\D^{\textup{\sffamily dt}}}
\newcommand{\DPT}{\D^{\textup{\sffamily pt}}}
\newcommand{\DCC}{\D^{\textup{\sffamily cc}}}
\newcommand{\RDT}{\R^{\textup{\sffamily dt}}}
\newcommand{\RPT}{\R^{\textup{\sffamily pt}}}
\newcommand{\RCC}{\R^{\textup{\sffamily cc}}}
\newmeasure{\disc}{disc}
\newmeasure{\bias}{bias}
\newcommand{\Dbar}{\overline{\D}}
\let\defaultS\S
\renewcommand{\S}{{\textup{\sffamily S}}\xspace}
\newcommand{\qbar}{\overline{q}}
\newcommand{\sq}{\mathop{\overline{sq}}}

\newcommand{\newfunction}[2]{\newcommand{#1}{{\textrm{\small #2}}}}
\newfunction{\FO}{\textsc{FO}}
\newfunction{\FPE}{\textsc{FPE}}
\newfunction{\FFO}{\textsc{FFO}}
\newfunction{\XOR}{\textsc{XOR}}
\newfunction{\NOR}{\textsc{NOR}}
\newfunction{\smNOR}{\textsc{\scriptsize NOR}}
\newfunction{\MAJ}{\textsc{MAJ}}
\newfunction{\OR}{\textsc{OR}}
\newfunction{\AND}{\textsc{AND}}
\newfunction{\ID}{\textsc{ID}}

\newcommand{\ZO}{\{0,1\}}
\newcommand{\ZS}{\{0,\star\}}
\newcommand{\ZSQ}{\{0,\star,\qmark\}}

\newcommand{\cT}{\mathcal{T}}
\newcommand{\cL}{\mathcal{L}}

\newcommand{\cR}{\mathcal{R}}
\newcommand{\cX}{\mathcal{X}}

\newcommand{\cU}{\mathcal{U}}
\newcommand{\cN}{\mathcal{N}}
\newcommand{\cO}{\mathcal{O}}

\newcommand{\cD}{\mathcal{D}}

\newcommand{\bmx}{\bm{x}}
\newcommand{\bmrho}{\bm{\rho}}
\renewcommand{\epsilon}{\varepsilon}
\newcommand{\qmark}{\mathord{?}} 

\newcommand{\buildlist}{\mathsf{BuildList}}
\newcommand{\ber}{\mathsf{Ber}}
\newcommand{\codim}{\textup{codim}}
\newcommand{\child}{\mathrm{child}}
\newcommand{\SPAN}{\textup{span}} 
\newcommand{\dotprod}[2]{\langle #1, #2 \rangle}
\newcommand{\dualnorm}[1]{\|\widehat{#1}\|_\infty^*}
\newcommand{\discreteD}{\cD^\times_{\mathbb{Z}}}
\newcommand{\infinitynorm}[1]{\| #1 \|_\infty}
\newcommand{\supp}{\mathrm{supp}}
\newcommand{\rank}{\mathrm{rank}}
\newcommand{\PATH}{\mathrm{path}} 
\newcommand{\Ext}[2]{\mathsf{Ext}_{#2}\mathopen{}\left(#1\mathclose{}\right)}
\newcommand{\pure}{\mathop{\textup{pure}}}
\newcommand{\indicator}[1]{\mathbbm{1}\left[#1\right]}
\newcommand{\col}{\mathrm{col}}
\newcommand{\err}{\mathrm{err}}
\newcommand{\spar}{\mathrm{spar}}
\newcommand{\dTV}{d_{\mathrm{TV}}}
\newcommand{\expected}{\mathop{\mathbb{E}}}
\newcommand{\expectedsub}[2]{\mathop{\mathbb{E}}_{#1}\left[#2\right]}
\begin{document}

\mbox{}\vspace{12mm}

\begin{center}

{\huge Direct Sums for Parity Decision Trees}\\[1cm] 

\large

\setlength\tabcolsep{1.2em}
\begin{tabular}{ccc}
Tyler Besselman&
Mika Göös&
Siyao Guo\\[-1mm]
\small\slshape NYU Shanghai&
\small\slshape EPFL&
\small\slshape NYU Shanghai
\end{tabular}
\vspace{1mm}
\begin{tabular}{cc}
Gilbert Maystre &
Weiqiang Yuan\\[-1mm]
\small\slshape EPFL&
\small\slshape EPFL
\end{tabular}

\vspace{2em}
	
\large
{\today}

\vspace{1em}
\end{center}

\begin{abstract}\noindent
Direct sum theorems state that the cost of solving $k$ instances of a problem is at least $\Omega(k)$ times the cost of solving a single instance. We prove the first such results in the randomised parity decision tree model. We show that a direct sum theorem holds whenever (1) the lower bound for parity decision trees is proved using the \emph{discrepancy method}; or (2) the lower bound is proved relative to a \emph{product distribution}.
\end{abstract}

\vspace{4em}

\setlength{\cftbeforesecskip}{0pt}
\renewcommand\cftsecfont{\mdseries}
\renewcommand{\cftsecpagefont}{\normalfont}
\renewcommand{\cftsecleader}{\cftdotfill{\cftdotsep}}
\setcounter{tocdepth}{1}

{\tableofcontents}

\thispagestyle{empty}
\setcounter{page}{0}
\newpage
\restoregeometry

\section{Introduction}
One of the most basic questions that can be asked for any model of computation is:
\begin{center}
\textit{How does the cost of computing $k$ independent instances scale with $k$?}
\end{center}
A \emph{direct sum} theorem states that if the cost of solving a single copy is $C$, then solving $k$ copies has cost at least~$\Omega(k\cdot C)$, which matches the trivial algorithm that solves the $k$ copies separately. Direct sums have been studied exhaustively for randomised query complexity $\RDT$, randomised communication complexity $\RCC$, and other concrete models of computation; see \cref{section:related_work} for prior work. In this work, we initiate the study of direct sum problems for randomised parity decision tree complexity~$\RPT$, a computational model sandwiched between the widely-studied $\RDT$ and $\RCC$.

\paragraph{Parity decision trees.}
Parity decision trees generalise the usual notion of decision trees by allowing \emph{parity queries}. To compute a function $f\colon \ZO^n \to \ZO$ on input $x\in\ZO^n$, a deterministic parity decision tree~$T$ performs queries of the form \emph{``what is $\dotprod{a}{x}$?''}\ where $a \in \ZO^n$ and $\dotprod{a}{x}\coloneqq{\sum_ia_ix_i\bmod 2}$. Once enough queries have been made, $T$ outputs $f(x)$. Parity decision trees are more powerful than ordinary decision trees: We have $\DPT(f) \leq \DDT(f)$ where $\DDT(f)$ (resp.\ $\DPT(f)$) denotes the (parity) decision tree complexity of $f$, defined as the least depth of a deterministic (parity) decision tree computing~$f$. On the other hand, the $n$-bit $\XOR$ function is an example where $\DDT(\XOR)=n$ while $\DPT(\XOR)=1$. We define a \emph{randomised} parity decision tree $\mathcal{T}$ as a distribution over deterministic parity trees $T\sim\mathcal{T}$. Then $\RPT_\epsilon(f)$ is defined as the worst-case depth (over both input and randomness of the tree) of the best randomised parity tree~$\mathcal{T}$ computing~$f$ with error $\epsilon$, that is, $\Pr[\mathcal{T}(x) \neq f(x)] \leq \epsilon$ for all $x$. As usual, we let $\RPT \coloneqq \smash{\RPT_{1/3}}$. To simplify notation, we drop the superscript {\sffamily pt} and write $\D=\DPT$ and $\R=\RPT$ for short.

\bigskip\noindent
Our main research question is now formulated as follows. Let $f^{k}\colon (\ZO^{n})^k \to \ZO^k$ denote the \emph{direct sum} function that takes $k$ instances $x \coloneqq (x^1, \dots, x^k)$ and returns the value of $f$ on each of them, $f^k(x) \coloneqq (f(x^1), \dots, f(x^k))$. We study the following question.
\begin{question} \label{main-q}
Do we have $\R(f^k)\geq \Omega(k)\cdot\R(f)$ for every function $f$?
\end{question}

We show two (incomparable) main results: We answer \cref{main-q} affirmatively when the randomised parity decision tree lower bound is proved using the \emph{discrepancy method} (\cref{sec:first-result}), or when the lower bound is proved relative to a \emph{product distribution} (\cref{sec:second-result}).

\subsection{First result: Direct sum for discrepancy}
\label{sec:first-result}

Discrepancy is one of the oldest-known methods for proving randomised communication lower bounds~\cite{Yao1983,Babai1986}. Let us tailor its definition to the setting of randomised parity trees. Thinking of $\ZO^n$ as the vector space $\mathbb{Z}_2^n$, consider some affine subspace $S \subseteq \ZO^n$ and a probability distribution $\mu$ over the inputs $\ZO^n$. The discrepancy of~$S$ measures how biased~$f$ is on~$S$. Namely, let $C^b_S \coloneqq \Pr_{\bmx \sim \mu}[f(\bmx) = b \,\wedge\, \bmx \in S]$. The difference~$\Delta_S \coloneqq |C^0_S - C^1_S|$ is called the bias of~$S$ under $\mu$. We define $\bias(f)$ as the minimum over $\mu$ of the maximum difference $\Delta_S$ an affine subspace can attain. Finally, the \emph{discrepancy bound} $\disc(f)$ is defined as $\log(1/\bias(f))$. As in communication complexity, it is not hard to see that $\R(f) \geq \Omega(\disc(f))$; see \cref{section:discrepancy} for details.

\begin{restatable}{theorem}{TheoremMaindisc}\label{theorem:main_disc}
We have $\R(f^k) \geq \Omega(k) \cdot \disc(f)$ for any function $f$.
\end{restatable}

In particular, if we have a function $f$ whose randomised parity decision tree complexity is equal to its discrepancy, $\R(f)=\Theta(\disc(f))$, then \cref{theorem:main_disc} shows $\R(f^k)\geq \Omega(k)\cdot \R(f)$ answering \cref{main-q} for that function. To prove~\cref{theorem:main_disc}, we first establish a particularly simple characterisation of~$\disc(f)$ that relies on affine spaces defined by a single constraint. We then prove a perfect direct sum (and even an $\XOR$ lemma) for discrepancy using Fourier analysis.

\subsection{Second result: Direct sum for product distributions} \label{sec:second-result}

The standard approach for proving randomised lower bounds is to use Yao's principle~\cite{Yao1977}, which states that $\R(f)=\max_\mu \D_{1/3}(f,\mu)$. Here $\D_\epsilon(f, \mu)$ is the \emph{distributional} $\epsilon$-error complexity of~$f$ defined as the least depth of a (deterministic) parity tree $T$ such that $\Pr_{\bmx \sim \mu}[T(\bmx) \neq f(\bmx)] \leq \epsilon$. We say that a distribution $\mu$ over $\ZO^n$ is \emph{product} if it can be written as the product of $n$ independent Bernoulli distributions. We define the best lower bound provable using a product distribution as
\[
\D^\times_\epsilon(f) \coloneqq \max_{\mu\text{ product}}\D_\epsilon(f,\mu)
\qquad\text{and}\qquad
\D^\times \coloneqq \D^\times_{1/3}.
\]
Our second result answers \cref{main-q} affirmatively (modulo logarithmic factors) whenever the randomised parity decision tree lower bound is proved relative to a product distribution.
\begin{restatable}{theorem}{TheoremMainDistributional}\label{theorem:main_distributional}
We have $\R(f^k) \,\geq\, \Omega(k/ \log n) \cdot \D^\times(f)$ for any $n$-bit function $f$.
\end{restatable}
We show moreover that the $O(\log n)$-factor loss in \cref{theorem:main_distributional} can be avoided when $\mu$ is the uniform distribution (or more generally any \emph{bounded-bias} distribution). One should compare this to the state-of-the-art in communication complexity, where the quantitatively best distributional direct sum results are also for product distributions and suffer logarithmic-factor losses~\cite{Jain2003,Barak2013}. 

To prove~\cref{theorem:main_distributional}, we introduce a new complexity measure tailored for product distributions, which we call \emph{skew complexity} $\S(f)$ and which we define precisely in~\cref{section:main_distributional}. We prove that this new measure admits a perfect direct sum theorem, $\S(f^k)=\Omega(k)\cdot\S(f)$, and that it characterises the measure $\D^\times$ up to an $O(\log n)$ factor. (We also show that the logarithmic loss is necessary for our approach: there is a function $f$ such that $\S(f)=O(1)$, even though $\D^\times(f) = \Theta(\log n)$.) We give a more in-depth technical overview in \cref{section:technical_overview}.

\paragraph{Comparison of main results.}

We also show that our two main results (\cref{theorem:main_disc,theorem:main_distributional}) are incomparable: For some functions~$f$, our first result gives a much stronger lower bound for $f^k$ than the second result---and vice versa. See \cref{section:disc-vs-Dprod} for the proof.
\begin{restatable}{lemma}{DiscVsDprod}\label{lemma:disc-vs-dprod}
The complexity measures $\disc$ and $\D^\times$ are incomparable:
\begin{enumerate}[noitemsep]
\item There is an $n$-bit function $f$ such that $\disc(f) = O(\log n)$ while $\D^\times(f) = \Theta(n)$.
\item There is an $n$-bit function $f$ such that $\disc(f)= \Theta(n)$ while $\D^\times(f) = O(1)$.
\end{enumerate}
\end{restatable}
\subsection{Related work}\label{section:related_work}
\paragraph{Parity decision trees.}
Even though the direct sum problem for parity decision trees has not been studied before, the model has been studied extensively.
Parity decision trees were first defined by Kushilevitz and Mansour~\cite{KM93} in the context of learning theory. Several prior works have studied their basic combinatorial properties~\cite{ZS10,ODonnell2014} as well as Fourier-analytic properties~\cite{Girish21,Girish23}, often with connections to the log-rank conjecture~\cite{Tsang2013,STV17,Sanyal19,Cheung2024Boolean,HHLO24,MS24}; see also the survey~\cite{KLMY21}. There are various lifting theorems involving parity decision trees: lifting from~$\DPT$ to $\DCC$~\cite{HHL18}, from $\DDT$ to $\DPT$~\cite{Chattopadhyay2023,Beame23,Alekseev24}, and from~$\RDT$ to $\RPT$~\cite{Shekhovtsov25,Byramji24}. These lifting theorems have played a central role in proving lower bounds for proof systems that can reason using parities~\cite{Itsykson20,Efremenko24,Filmus24,Bhattacharya24,Chattopadhyay24,Alekseev24b}.
\paragraph{Decision trees.}
In the decision tree model with classical queries, a deterministic direct sum theorem, $\DDT(f^k) = k \cdot \DDT(f)$, and even the stronger \emph{composition theorem}, $\DDT(g\circ f^k)=\DDT(g) \cdot \DDT(f)$, are easy to show by combining adversary strategies~\cite{Savicky2002}. In the randomised case, an optimal direct sum result, $\RDT(f^k)\geq \Omega(k)\cdot\RDT(f)$, is known~\cite{Klauck2007,Jain2010,Drucker12}. Whether a composition theorem holds for randomised query complexity, $\RDT(g\circ f^k)\geq \Omega(\RDT(g) \cdot \RDT(f))$ (for total $g$ and $f$), is a major open problem~\cite{BenDavid2018,BenDavid2020,BenDavid2022,BenDavid2023,Sanyal24}. In the randomised setting, it is possible that the direct sum problem~$f^k$ requires strictly more than $\Theta(k)\cdot\RDT(f)$ queries: if one wants to succeed in computing all~$k$ copies with probability $\geq 2/3$, then a naive application of the union bound would require each copy to have error~$\ll 1/k$. Results stating that one sometimes has $\RDT(f^k)\geq \omega(k)\cdot\RDT(f)$ are called ``strong'' direct sum theorems~\cite{Blais2019,Blanc2024} and they sometimes hold even for composed functions~\cite{BenDavid20random,Brody2023,Goos2021}.
\paragraph{Communication complexity.}
The direct sum question for deterministic communication complexity was posed in~\cite{Tomas1995} and it remains a notoriously difficult open problem~\cite{Iyer2024}. By contrast, in the randomised setting, the direct sum problem is characterised by \emph{information complexity}~\cite{Braverman14}, which has inspired a line of works too numerous to cite here; see~\cite[\defaultS1.1]{Iyer24stoc} for an up-to-date overview. One of the key findings is that a direct sum for communication protocols is \emph{false} in full generality in the distributional setting~\cite{Ganor16,Rao18}. We leave open the intriguing possibility that the information complexity approach can be adapted to parity decision trees. Historically, one of the first direct sum theorems proved for randomised communication was for the discrepancy bound~\cite{Shaltiel2003,Lee2008} (analogously to our \cref{theorem:main_disc}). Here, discrepancy is known to be equivalent to the $\gamma_2$-norm~\cite{Linial2008}. We also mention that a near-optimal direct sum theorem holds for product distributions~\cite{Barak2013} (analogously to our \cref{theorem:main_distributional}).
\subsection{Open question: Deterministic direct sum}
The main question left open by our work is \cref{main-q}, namely, whether $\R=\RPT$ admits a direct sum theorem for all functions $f$. However, we would also like to highlight the analogous question in the deterministic case $\D=\DPT$. As discussed above, this is a long-standing open problem in the case of deterministic communication complexity $\DCC$. The best results so far are:
\begin{enumerate}[noitemsep]
    \item $\DCC(f^k) \geq \tilde{\Omega}(k) \cdot \DCC(f)^{1/2}$ as proved in~\cite{Tomas1995}.
    \item $\DCC(f^k) \geq \tilde{\Omega}(k) \cdot \DCC(f)/\log \rank(f)$ as proved in~\cite{Iyer2024}.
\end{enumerate}
We observe in \cref{app:deterministic-case} that both approaches have analogues in the parity setting.
\begin{restatable}{theorem}{DeterministicDirectSums}\label{theorem:deterministic_direct_sums}
For any function $f$ and $k \geq 1$,
\begin{enumerate}[noitemsep]
    \item $\D(f^k) \geq k \cdot \D(f)^{1/2}$, \label{item:deterministic_1}
    \item $\D(f^k) \geq k \cdot \D(f)/\log\spar(f)$. \label{item:deterministic_2}
\end{enumerate}
\end{restatable}
We leave it as an open question whether a perfect direct sum theorem holds for deterministic parity decision trees.  We think one should attack this problem before addressing the (presumably much harder) problem for deterministic communication complexity.
\section{Technical overview}\label{section:technical_overview}
We focus here on our second main result in \cref{theorem:main_distributional} stating that $\R(f^k) \geq \Omega(k/\log n) \cdot \D^\times(f)$ and which is technically the much more involved theorem. Our main technical result is the following direct sum result for distributional complexity. Here $\mu^k\coloneqq \mu\times\cdots\times\mu$ ($k$ times).
\begin{restatable}{theorem}{theoremMainTechnical}\label{theorem:main_technical}
There exists a universal constant $C$ such that the following holds.
For any $f\colon \ZO^n \to \ZO$, product distribution $\mu$ over $\{0,1\}^n$, and $k \geq 1$,
\begin{equation*}
\D_\epsilon(f^k, \mu^k) \geq \Omega\left(\frac{k\delta}{\log (n/\delta)}\right) \cdot (\D_{\epsilon + \delta}(f, \mu)-C\cdot\log(n/\delta)) \qquad  \forall \epsilon,\delta \geq 0.
\end{equation*}
\end{restatable}
When $\D^{\times}(f) \geq 6C \cdot \log n$, \cref{theorem:main_distributional} follows by taking $\epsilon = \delta = 1/6$. Indeed, let $\mu$ be the distribution achieving the maximum for $\D^\times$. Using the easy direction of the minimax principle:
\[
\R(f^k)\geq \Omega(1) \cdot \D_{1/6}(f^k, \mu^k) \geq \Omega(k/\log n) \cdot \D_{1/3}(f, \mu) = \Omega(k/\log n)\cdot \D^\times(f).
\]
The remaining case $\D^{\times}(f) \leq 6C \cdot \log(n)$ is handled separately using ad-hoc methods in \cref{lemma:low_distributional}. We now give an overview of the proof of \cref{theorem:main_technical}.

\paragraph{Warm-up: Uniform distribution.}
We showcase the basic proof technique by sketching the proof in the simple case where $\mu$ is the uniform distribution. 
Fix an $n$-bit function $f$ and let $\mathcal{U}$ be the uniform distribution over $\ZO^n$. In the uniform (and more generally in the \emph{bounded-bias}) case, we are actually able to avoid the $\log n$ additive/factor loss and obtain, for all $k \geq 1$,
\begin{equation}\label{equation:perfect_direct_sum_uniform}
\D_\epsilon(f^k, \mathcal{U}^k) \geq \Omega(k\delta) \cdot \D_{\epsilon + \delta}(f, \mathcal{U}) \qquad \forall \delta\geq 0.
\end{equation}
Fix a decision tree $T$ of depth $d$ computing $k$ copies of $f$ with error at most $\epsilon$ when $\bmx \sim \mathcal{U}^k$. We show how to extract a tree $T^*$ that computes a single copy $\bm{y} \sim \mathcal{U}$ with error at most $\epsilon + \delta$ and depth $\leq O(d/k\delta)$. Leaves of $T$ correspond to affine subspaces of $(\ZO^n)^k$ of codimension $\leq d$. More generally, one can associate with any node $v$ of $T$ the set $C_v = \{w_1,\, \dots, w_{d(v)}\}$ of linear constraints that led to the node ($d(v)$ is the depth of the node $v$; the root is at level 0) and the vector $b \in \ZO^k$ of desired values. The set of inputs $S_v$ that reach node $v$ is then given by $S_v \coloneqq \{x \in (\ZO^n)^k:\, \dotprod{w_j}{x} = b_j,\ \forall j \in [d(v)]\}$.

Of relevance here are the \emph{pure constraints} one can extract from $C_v$. A pure constraint for copy $i \in [k]$ is some $w \in (\ZO^n)^k$ such that $w^j \neq 0^n$ if and only if $j = i$. To be more precise, the number of pure queries that can be extracted for query $i$ at node $v$ is defined with:
\begin{equation*}
\pure\nolimits_i(C_v) \coloneqq \dim(\SPAN(C_v) \cap W_i) \quad\text{where}\quad W_i \coloneqq \big\lbrace w \in (\ZO^n)^k:\, w^j = 0^n,\ \forall j \neq i\big\rbrace.
\end{equation*}
We describe next two illustrative examples when there are $k=2$ copies.
\begin{enumerate}
\item Node $v$ corresponds to constraints ``$x^1_1 + x^2_1 = 0$'' and ``$x^2_1 = 1$''. Then, $\pure_1(C_v) = 1$ as it is possible to extract the pure parity constraint $x^1_1 = 1$ by adding the two constraints. In the same vein, $\pure_2(C_v) = 1$.
\item Node $v$ corresponds to constraints ``$x^1_1 + x^2_1 = 0$'' and ``$x^2_1 + x^2_2 = 1$''. Then, $\pure_1(C_v) = 0$ as it not possible to extract a pure constraint for the first copy.
\end{enumerate}
\begin{observation}\label{observation:rank_sum}
For any node $v$, we have $d(v) \geq \sum_{i \in [k]} \pure_i(C_v)$.
\end{observation}
As the second example highlights, it is possible for the inequality to be strict. This is a notable difference with classical decision trees: for any subcube $C \in (\{0,\, 1,\, *\}^n)^k$, the sum of fixed bits of each copy is the total number of fixed bits in $C$.
\paragraph{Where to plant $y$?}
The overarching idea of our result is that under the uniform distribution, \emph{queries that increase the pure rank for a copy are the only ones that bring usable information.} It is thus enough to find a copy with low expected pure rank in $T$ and plant the real instance $y$ there. To make this precise, taking the expectation over leaves of $T$ when $\bmx \sim \mathcal{U}$ with \cref{observation:rank_sum} implies the existence of some copy $i \in [k]$ with low expected pure rank:
\begin{equation*}
\expected\nolimits_{\bmx \sim \mathcal{U}^k}[\pure\nolimits_i(C_{\ell(\bmx)})] \leq O(d/k). 
\end{equation*} 
Let us fix this advantageous copy to be $i = 1$. On input $y \in \ZO^n$ we run the tree $T$ with $y$ planted as $x^1$ and delay actual querying of bits of $y$ as much as possible. Suppose that the process has reached node $v$ with constraint set $C_v$ and there is a new parity query $w$ to be answered. If~$w \in \SPAN(C_v)$, the answer to that query can be found (an optimised tree would not do such a query). If $w \notin \SPAN(C_v)$, we say that $w$ is \emph{critical} for $C_v$ if it would increase the pure rank for the first copy $\pure_1(C_v \cup \{w\}) > \pure_1(C_v)$. If $w$ is critical, there is no way to avoid making a parity query to the real input $y$ and our algorithm does it. If $w$ is not critical, it is however enough to answer with a uniform bit (that is, move to a random child of $v$ in $T$) without querying $y$ at all.

To see this, further split $w=w^1w^{-1}$, where $w^1 \in \ZO^n$ is the constraint for the first copy and $w^{-1} \in (\ZO^n)^{k-1}$ is the constraint for the rest of the copies. If $w$ has $\pure_1(C_v \cup \{w\}) = \pure_1(C_v)$ and $w \notin \SPAN(C_v)$, it must be that $0^nw^{-1} \notin \SPAN(C_v)$. Since $\bmx^{-1}$ is drawn from the uniform distribution we thus have for any fixed $y$ consistent with $S_v$:
\begin{equation}\label{equation:specifics_uniform}
\Pr_{\bmx^{-1}}\left[\dotprod{w}{y\bmx^{-1}} = 0 \,|\, (y, \bmx^{-1}) \in S_v\right] = \Pr_{\bmx^{-1}}\left[\dotprod{w^{-1}}{\bmx^{-1}} = \dotprod{w^1}{y} \,|\, (y, \bmx^{-1}) \in S_v\right] = \frac{1}{2}.
\end{equation}
\paragraph{Correctness and efficiency.}
Let us call the above randomised tree solving one copy as $\mathcal{T}$. Correctness can be argued by showing that the distribution of leaves attained in the process for $\bm{y} \sim \mathcal{U}$ is the same as the distribution of leaves attained by $\bmx \sim \mathcal{U}^k$ in $T$. On the other hand, $\mathcal{T}$ has expected depth $O(d/k)$ as a real query to $\bm{y}$ is only ever made $\pure_i(C_\ell)$ times for each leaf $\ell$. In conclusion, $\mathcal{T}$ has the following guarantees:
\begin{enumerate}
\item $\Pr_{\bm{y} \sim \mathcal{U}, \bm{T} \sim \mathcal{T}}[\bm{T}(\bm{y}) \neq f(\bm{y})] \leq \epsilon$.
\item $\expected_{\bm{y} \sim \mathcal{U}, \bm{T} \sim \mathcal{T}}[\text{\#queries}(\bm{T}, \bm{y})] \leq d/k$.
\end{enumerate}
Using Markov inequality, it is possible to derandomise $\mathcal{T}$ to get a deterministic parity tree $T^*$ solving~$f$ with a worst-case guarantee instead of an average-case one. This step introduces a parameter $\delta$ controlling a trade-off between cost and error and yields the desired result \eqref{equation:perfect_direct_sum_uniform}.
\subsection{Beyond uniform: The skew measure}
Observe that \eqref{equation:specifics_uniform} can fail badly for non-uniform $\mu$. As an illustrative example suppose that two random bits $\bm{a},\bm{b}$ are generated with $\bm{a} \sim \ber(1/2)$ and $\bm{b} \sim \ber(1/8)$. The constraint $\bm{a} \oplus \bm{b} = 1$ is not pure from the point of view of $\bm{a}$. However, since $\bm{b}$ is skewed towards being 0, the realisation of the constraint gives information about $\bm{a}$: $\Pr[\bm{a} = 0 \,|\, \bm{a} + \bm{b} = 1] = 1/8 \ll 1/2$. Thus, it seems one needs to query $\bm{a}$ to answer the query $\bm{a} + \bm{b}$ even though the query is not critical for $\bm{a}$!

To circumvent this, we introduce the \emph{skew} measure. This new measure is built around the observation that each bit of an input $\bmx \sim \mu$ can be sampled independently in two steps. Indeed, the following process is equivalent to $\ber(1/8)$: 
\begin{enumerate}
    \item Let $\bmrho \in \{0,\, \star\}$ be `$0$' with probability 3/4 and $\star$ with probability 1/4.
    \item If $\bmrho = 0$, return `$0$', else return a sample $\ber(1/2)$.
\end{enumerate}
Note that if we are ``lucky'' and $\bmrho = \star$, we are back in the uniform case and \eqref{equation:specifics_uniform} holds again. If not, we have somehow pre-emptively fixed the return bit to value $0$. The skew measure explicitly splits product distributions into a \emph{random partial fixing} $\bmrho$ followed by a uniform distribution over unfixed bits of $\bmrho$. A tree computing in this model gets help from $\bmrho$ because $\bmrho$ reduces the complexity of the function. When those bits are unfixed, it is on the other hand easier to analyse the behaviour of the tree as it is the uniform case again.

In \cref{section:Direct_Sum_for_S,section:Compression}, we show a perfect direct sum for the skew measure and that perhaps surprisingly, this new measure is only a $\log n$-factor away from $\D^\times$.
\section{Direct sum for \texorpdfstring{$\disc$}{disc}} \label{section:discrepancy}
The goal of this section is to prove \cref{theorem:main_disc}, restated here for convenience.
\TheoremMaindisc*
Let us start by defining discrepancy formally. We denote by $\mathcal{S}_n$ the set of all affine subspaces of $\ZO^n$ and $\mathcal{O}_n \subseteq \mathcal{S}_n$ the set of affine subspaces of codimension 1. Note that all spaces $S \in \mathcal{O}_n$ can be written as $S = \{x \in \ZO^n: \dotprod{a}{x} = b\}$ for some $a \in \ZO^n$ and $b \in \ZO$.
\begin{definition}\label{def:discrepancy}
Let $f\colon \ZO^n\to\ZO$ be a boolean function and $\mu$ be a distribution over $\ZO^n$. The (parity) discrepancy of $f$ with respect to $\mu$ is defined as:
\[
\disc(f, \mu) \coloneqq -\log \max_{S \in \mathcal{S}^n} \bias(f, \mu, S) \quad \text{where} \quad \bias(f, \mu, S) \coloneqq \left\vert \sum\nolimits_{x \in S} (-1)^{f(x)} \mu(x) \right\vert.
\]
The (parity) discrepancy of $f$ is $\disc(f) \coloneqq \max_\mu \disc(f, \mu)$ where $\mu$ ranges over all distributions.
\end{definition}
Observe that $\disc(f) \geq 1$ for all non-constant $f$ and by standard arguments, $\R(f) \geq \disc(f)$ (see \cref{lemma:pR_vs_disc}). Using the latter, the only thing left to get \cref{theorem:main_disc} is to prove a direct sum result for discrepancy. We do this in a very strong way by actually establishing an $\XOR$ lemma for $\disc$. Let $f^{\oplus k}$ denote the function that takes $k$ instance and aggregates their result under $f$ using $\XOR$, so that $f^{\oplus k}(x^1, \dots, x^k) \coloneqq f(x^1) \oplus \, \cdots \, \oplus f(x^ k)$.
\begin{lemma}\label{lemma:direct_sum_disc}
For any function $f$, distribution $\mu$ and $k \geq 1$,
\[
k \cdot \disc(f, \mu) \geq \disc(f^{\oplus k}, \mu^k) \geq k \cdot \big(\disc(f, \mu) - 1\big).
\]
\end{lemma}
This result is the strongest possible. Indeed, we cannot omit the ``$-1$'' on the right because of the counterexample $f\coloneqq\XOR$: we have $\disc(f^{\oplus k},\mu^k)\leq 1$ for any distribution $\mu$. In \cref{section:appendix_disc_distribution_free} we revisit this $\XOR$ lemma and show that it also holds in the distribution-free setting, with $\disc(f^{\oplus k}) \approx k \cdot \disc(f)$. As a final comment, we note that it is easier to work with $f^{\oplus k}$ instead of $f^k$ in the discrepancy setting, as it is somewhat tedious to define discrepancy for multi-valued functions. Before formally proving \cref{lemma:direct_sum_disc}, we show how it is used to prove the main result \cref{theorem:main_disc}.
\begin{proof}[Proof of \cref{theorem:main_disc}]
Any decision tree computing $f^k$ can be converted to a decision tree computing $f^{\oplus k}$. This is achieved by replacing the label $y \in \ZO^k$ of each leaf by its parity $\dotprod{y}{1^k}$. This operation does not increase the error probability or cost and so, using the easy direction of Yao's principle:
\begin{align*}
\R(f^k) &\geq \max_{\mu} \D(f^k, \mu^k, 1/3) &\text{(\cref{lemma:yao_minimax})}\\
&\geq \max_\mu \D(f^{\oplus k}, \mu^k, 1/3)\\
&\geq \max_\mu \disc(f^{\oplus k}, \mu^k) - \log_2(3) &\text{(\cref{lemma:pR_vs_disc})}\\
&\geq k \cdot \max_\mu (\disc(f, \mu) - 1) - \log_2(3) &\text{(\cref{lemma:direct_sum_disc})}\\
&\geq k \cdot (\disc(f) - 1) - \log_2(3).
\end{align*}
If $\disc(f) \geq 10$, then the string of inequalities yields $k \cdot(\disc(f) - 1) - \log_2(3) \geq k \cdot\disc(f)/10$. If $f$ is constant, the claim is vacuously true. Finally, we show that for any non-constant $f$, $\R(f^k) \geq k - \log(3/2)$ which completes the claim. Indeed, if $\disc(f) \leq 10$, then $k - \log(3/2) \geq k \cdot \disc(f) / 100$.

To this end, let $f$ be a non-constant function and $\mu$ a distribution over $\ZO^n$ which is balanced over $0$-inputs and $1$-inputs, i.e. $\mu(f^{-1}(0)) = \mu(f^{-1}(1)) = 1/2$. Let $T$ be the best deterministic parity decision tree for $\D_{1/3}(f, \mu)$ and suppose toward contradiction that it has strictly less than $L \coloneqq 2^k \cdot (2/3)$ leaves. Let $G \subseteq \ZO^n$ be the set of solutions which appear as a label on a leaf of $T$. We have $|G| < L$ and since $\mu$ is balanced, any solution $y \in \ZO^k$ is equally likely so that:
\[
\Pr_{\bmx \sim \mu^k}[T(\bmx) = f^{k}(\bmx)] \leq \Pr_{\bmx \sim \mu^k}[f^k(\bmx) \in G] \leq |G| \cdot 2^{-k} < 2/3.
\]
Thus, $T$ errs with probability $> 1/3$: a contradiction.
\end{proof}
We now proceed to prove \cref{lemma:direct_sum_disc} in three steps.
\subsection{Step 1: Characterisation of discrepancy}
Much like discrepancy for communication protocols can be characterised by the $\gamma_2$-norm of the communication matrix \cite{Shaltiel2003,Linial2008}, we show that the parity discrepancy of $f$ on $\mu$ is characterised by the $L_\infty$-norm of the Fourier transform of a related function $F_\mu$. This characterisation has two purposes. First, proving an $\XOR$ lemma requires exploring all the possible ways for the $k$ copies to sum to 1. This kind of convolution operation is greatly simplified in the Fourier domain, where it simply corresponds to standard multiplication. Second, the characterisation is also quite convenient to prove lower bounds on $\disc(f, \mu)$ (which we do in \cref{section:disc-vs-Dprod,section:separation}): it shows that maximum bias is (almost) attained for affine spaces of codimension 1 already.
\paragraph{The function $F_\mu$.}
We relate a real-valued boolean function $F\colon \ZO^n \to \mathbb{R}$ with its Fourier transform $\widehat{F}: \ZO^n \to \mathbb{R}$ using the usual basis:
\begin{align*}
    \forall z\in \ZO^n, \quad \widehat{F}(z) &\coloneqq \sum\nolimits_{x \in \ZO^n} F(x) \cdot (-1)^{\dotprod{x}{z}} \cdot 2^{-n}; &\textup{\small [Fourier transform]}\\
    \forall x \in \ZO^n, \quad F(x) &\coloneqq \sum\nolimits_{z \in \ZO^n} \widehat{F}(z) \cdot (-1)^{\dotprod{z}{x}}. &\textup{\small [Inverse Fourier transform]}
\end{align*}
See also \cite{ODonnell2014Book} for more background on Fourier analysis. We use $\infinitynorm{\widehat{F}}$ to denote the maximum absolute value of a Fourier coefficient of $F$. To analyze $\disc(f, \mu)$, we introduce an associated function $F_\mu\colon \ZO^n \to \mathbb{R}$ defined by $F_\mu(x) \coloneqq (-1)^{f(x)} \cdot \mu(x) \cdot 2^n$ and prove the following characterisation.%
\begin{lemma}\label{lemma:disc_characterisation}
For every function $f\colon\ZO^n \to \ZO$ and distribution $\mu$ over $\ZO^n$:
\[
\max_{S \in \mathcal{O}_n} \bias(f, \mu, S) \leq \max_{S \in \mathcal{S}_n}\bias(f, \mu, S) \leq \infinitynorm{\widehat{F_\mu}} \leq 2 \cdot \max_{S \in \mathcal{O}_n} \bias(f, \mu, S).
\]
\end{lemma}
\begin{proof}
The first inequality holds immediately because $\mathcal{O}_n \subseteq \mathcal{S}_n$. For the second,  fix a maximizing $S \in \mathcal{S}^n$. Suppose that $\codim(S) = d$ and fix its constraints $a_j \in \ZO^n$ and $b_j \in \ZO$ for $j \in [d]$ so that $S = \{x \in \ZO^n: \dotprod{a_j}{x} = b_j \,\, \forall j \in [d]\}$. Observe that the vectors $\{a_j\}_{j \in [d]}$ are linearly independent. Let $\Phi \coloneqq \sum\nolimits_{x\in S}(-1)^{f(x)} \mu(x)$ so that $\bias(f, \mu, S) = |\Phi|$ and observe that
\begin{equation*}
\Phi = 2^{-n} \cdot \sum_{x \in S} F_\mu(x) = 2^{-n} \cdot \sum_{x \in S} \sum_{z \in \ZO^n} \widehat{F}_\mu(z) (-1)^{\dotprod{z}{x}} = 2^{-n} \cdot \sum_{z \in \ZO^n} \widehat{F}_\mu(z) \sum_{x \in S}  (-1)^{\dotprod{z}{x}}.
\end{equation*}
We focus on analysing terms $T_z \coloneqq \sum_{x \in S}(-1)^{\dotprod{z}{x}}$. Let $V \coloneqq \SPAN\{a_1,\, \dots,\, a_d\}$ and observe that whenever $z \in V$, $|T_z| = |S|$. Indeed, if $\beta_1, \dots, \beta_d \in \ZO$ is a linear combination of $z$ in $V$:
\[
T_z = \sum_{x \in S} (-1)^{\dotprod{z}{x}} = \sum_{x \in S} \prod_{j \in [d]} (-1)^{\beta_j\dotprod{a_j}{x}} = \sum_{x \in S} (-1)^{\sum_j \beta_j b_j} = |S| \cdot (-1)^{\sum_{j} \beta_j b_j}.
\]
On the other hand, $T_z = 0$ for all $z \notin V$. Indeed, Letting $S^b = S \cap \{x \in \ZO^n: \dotprod{x}{z} = b\}$ we have $T_z = |S^0| - |S^1|$. Because $z \notin V$, the constraint $\dotprod{x}{z} = b$ splits $S$ in half and thus $|S^0| = |S^1| = |S|/2$. Factoring in those observations, we get:
\begin{equation*}
|\Phi| = 2^{-n} \cdot \Big\lvert \sum\nolimits_{z \in \ZO^n} \widehat{F_\mu}(z) \cdot T_z \Big\rvert \leq 2^{-n} \cdot |S| \cdot \sum\nolimits_{z \in V} \left\vert \widehat{F_\mu}(z) \right\vert \leq 2^{-n} \cdot |S| \cdot |V| \cdot \infinitynorm{\widehat{F_\mu}}.
\end{equation*}
Recall that $S$ has codimension $d$ and as such $|S| = 2^{n - d}$ and $|V| = 2^d$, implying the desired inequality $\bias(f, \mu, S) \leq \infinitynorm{\widehat{F_\mu}}$. We now prove the third inequality of the lemma. Fix any maximum Fourier coefficient $y^\star \in \ZO^n$ and observe:
\begin{align*}
\infinitynorm{\widehat{F_\mu}} = |\widehat{F_\mu}(y^\star)| = \Big\lvert \sum\nolimits_{x \in \ZO^n} F_\mu(x) \cdot (-1)^{\dotprod{x}{y^\star}} \cdot 2^{-n}\Big\rvert \leq 2 \cdot \max_{b \in \ZO} \Big\lvert \sum\nolimits_{x: \dotprod{x}{y}=b} F_\mu(x) \cdot 2^{-n}\Big\rvert.
\end{align*}
Fix the maximizing argument to $b^\star$ and define $S^\star \coloneqq \{x \in \ZO^n: \dotprod{x}{y^\star} = b^\star\}$. Note that $S^\star \in \mathcal{O}_n$ and as such:
\begin{equation*}
\infinitynorm{\widehat{F_\mu}} \leq 2 \cdot \left\vert \sum\nolimits_{x \in S^\star} (-1)^{f(x)}\mu(x) \right\vert \leq 2 \cdot \max_{S \in \mathcal{O}_n} \bias(f, \mu, S).\qedhere
\end{equation*}
\end{proof}
\subsection{Step 2: Direct sum for the maximum Fourier coefficient}
The outer-product of functions $F,G: \ZO^n \to \mathbb{R}$ is defined as the function $F \otimes G: \ZO^{2n} \to \mathbb{R}$ with $(F \otimes G)(x^1, x^2) \coloneqq F(x^1) \cdot G(x^2)$. Next is a direct sum result for its max Fourier coefficient.
\begin{claim}\label{claim:fourier_direct_sum}
For any $F,G:\ZO^n \to \mathbb{R}$, $\infinitynorm{\widehat{F \otimes G}} = \infinitynorm{\widehat{F}} \cdot \infinitynorm{\widehat{G}}$.
\end{claim}
\begin{proof}
Let $H = F \otimes G$; for any $z^1,z^2 \in \ZO^n$, the definition of Fourier transform implies
\begin{align*}
\widehat{H}(z^1, z^2) &= 2^{-2n} \cdot \sum\nolimits_{x^1,x^2 \in \ZO^n} H(x^1, x^2) \cdot (-1)^{\dotprod{x^1x^2}{z^1z^2}}\\
&= 2^{-2n} \cdot \sum\nolimits_{x^1,x^2 \in \ZO^n} F(x^1) \cdot G(x^2) \cdot (-1)^{\dotprod{x^1}{z^1}} \cdot  (-1)^{\dotprod{x^2}{z^2}}\\
&= \widehat{F}(z^1) \cdot \widehat{G}(z^2).  
\end{align*}
From this, the equivalence is immediate:
\begin{equation*}
\infinitynorm{\widehat{H}} = \max_{z^1, z^2} |\widehat{H}(z^1, z^2)| = \max_{z^1, z^2} |\widehat{F}(z^1)| \cdot |\widehat{G}(z^2)| = \infinitynorm{\widehat{F}} \cdot \infinitynorm{\widehat{G}}. \qedhere
\end{equation*}
\end{proof}
\subsection{Step 3: Conclusion}
We tie together \cref{lemma:disc_characterisation,claim:fourier_direct_sum} and prove \cref{lemma:direct_sum_disc}.
\begin{proof}[Proof of \cref{lemma:direct_sum_disc}]
Let $H\colon (\ZO^{n})^k \to \mathbb{R}$ be the function associated with $f^{\oplus k}$ and $\mu^k$ in \cref{lemma:disc_characterisation}. It is possible to express $H$ as the $k$-fold outer-product of $F_\mu$: $H = F_\mu \otimes \cdots \otimes F_\mu$. Indeed, for $x \in (\ZO^n)^k$, we have:
\begin{equation*}
H(x) = 2^{-kn} \cdot (-1)^{f^{\oplus k}(x)}\mu^k(x) = \prod\nolimits_{i \in [k]} 2^{-n} (-1)^{f(x^i)}\mu(x^i) = \prod\nolimits_{i \in [k]} F_\mu(x^i).
\end{equation*}
Thus, using the characterisation of \cref{lemma:disc_characterisation} and \cref{claim:fourier_direct_sum} $k$ times: 
\begin{equation*}
\max_{S \in \mathcal{S}_{kn}} \bias(f^{\oplus k}, \mu^k, S) \leq \infinitynorm{\widehat{H}} = \left( \infinitynorm{\widehat{F_\mu}}\right)^k \leq 2^k \cdot \left(\max_{S \in \mathcal{S}_n} \bias(f, \mu, S)\right)^k.
\end{equation*}
The $\XOR$-lemma $\disc(f^{\oplus k}, \mu^k) \geq k \cdot (\disc(f, \mu) - 1)$ follows directly. We now show the other direction, $\disc(f^{\oplus k}, \mu^k) \leq k \cdot \disc(f, \mu)$. To do so, fix some $S \in \mathcal{S}_n$ maximizing $\bias(f, \mu, S)$ and define $T \in \mathcal{S}_{kn}$ which is concatenation of $k$ copies of $S$. Formally:
\begin{equation*}
T = \big\lbrace x \in (\ZO^n)^k: x^i \in S \quad \forall i \in [k] \big\rbrace.
\end{equation*}
Now, it is easy to check that $\bias(f^{\oplus k}, \mu^k, T) = \bias(f, \mu, S)^k$ and the claim follows.
\end{proof}
\section{Direct sum for \texorpdfstring{$\D^\times$}{D product} part I: proof organisation} \label{section:main_distributional}
The goal of this section is to prepare the ground for a proof of our main technical contribution: a direct sum for parity trees in the distributional setting (restated below).
\theoremMainTechnical*
As stated in \cref{section:technical_overview}, this is sufficient to prove \cref{theorem:main_distributional} whenever $\D^\times(f) \geq 6C \cdot \log(n)$. The remaining case $\D^\times(f) \leq 6C \cdot \log(n)$ is proved in \cref{lemma:low_distributional} in \cref{subsection:ommited_case}. We thus focus on proving \cref{theorem:main_technical} in the next two sections (this section is devoted to introducing the necessary definitions and technical lemmas).
\subsection{Two strengthenings of \cref{theorem:main_technical}}
For technical convenience, we study distributional complexity for \emph{randomised} trees. For a deterministic parity tree $T$ we let $q(T, x)$ be the number of queries made by $T$ on input $x$. If $\mathcal{T}$ is a randomised tree and $\mu$ is a distribution, we define $\qbar(\mathcal{T}, \mu)$ and $\err_f(\mathcal{T}, \mu)$ in the natural way with:
\begin{equation*}
\qbar(\mathcal{T}, \mu) \coloneqq \expected\nolimits_{\substack{\bm{T} \sim \mathcal{T}\\\bmx \sim \mu}}[q(\bm{T}, \bmx)] \quad\textup{and}\quad \err_f(\mathcal{T}, \mu) \coloneqq  \Pr\nolimits_{\substack{\bm{T} \sim \mathcal{T}\\ \bmx \sim \mu}}[\bm{T}(\bmx) \neq f(\bmx)].
\end{equation*}
Finally, we define $\Dbar_\epsilon(f, \mu) = \min_{\mathcal{T}} \{\qbar(\mathcal{T}, \mu): \err_f(\mathcal{T} , \mu) \leq \epsilon\}$. It is clear that $\Dbar_\epsilon(f, \mu) \leq \D_\epsilon(f, \mu)$ but a converse result is more complicated, as the derandomisation can increase both the error and the depth simultaneously.
\begin{restatable}{claim}{claimpruningPDT}\label{claim:pruning_PDTs}
For any $f\colon \ZO^n \to \ZO$, $\mu$ over $\ZO^n$ and $\epsilon, \delta \geq 0$, $\D_{\epsilon+\delta}(f,\mu) \le \Dbar_\epsilon(f,\mu)/\delta$.
\end{restatable}
We delay a proof of this folklore fact to \cref{section:appendix_fact_on_trees}.
We also refer readers to~\cite{Jain2010} which proves the analogue for ordinary decision trees.
With this tool in hand, we can reduce \cref{theorem:main_technical} to the following theorem.
\begin{theorem}\label{theorem:direct_sum_for_avgD}
There exists a universal constant $C$ such that the following holds.
For any $f\colon \ZO^n\to \ZO$, product distribution $\mu$, and $k \geq 1$,
\[
\quad\Dbar_{\epsilon}(f^k,\mu^k) \geq \Omega\big(k/\log(n/\gamma)\big)\cdot (\Dbar_{\epsilon+\gamma}(f,\mu)-C\cdot \log(n/\gamma)) \quad \forall \gamma \in (0, 1/n).
\]
\end{theorem}
\begin{definition}\label{definition:lambda_bounded}
We say that a product distribution $\mu$ over $\ZO^n$ is $\lambda$-bounded for some $\lambda \in (0, 1]$ if $\Pr_{\bmx \sim \mu}[\bmx_i = 1] \in [\lambda/2, 1 - \lambda/2]$ for every $i \in [n]$.
\end{definition}
In the next sections, we also show the following qualitative improvement over \cref{theorem:direct_sum_for_avgD} for bounded distributions.
\begin{theorem}\label{theorem:direct_sum_for_constant_bounded_avgD}
For any $f\colon\ZO^n\to \ZO$, $\lambda$-bounded distribution $\mu$ and $k \geq 1$,
\[
\Dbar_{\epsilon}(f^k,\mu^k) \geq \Omega\left(k\lambda\right) \cdot \Dbar_{\epsilon}(f,\mu).
\]
\end{theorem}
Let us highlight the difference between \cref{theorem:direct_sum_for_avgD} and \cref{theorem:direct_sum_for_constant_bounded_avgD}: the latter is free from both the $\log n$ factor and the extra error $\gamma$. This theorem is especially interesting when the hard distribution for the function at hand (e.g.\ $\MAJ$) is close to the uniform one.
\subsection{The Skew measure}
For the rest of this paper, we let $\cU$ be the uniform distribution. Let $\mu$ be a distribution over $\ZO^n$ and $S \subseteq \ZO^n$. We use $\mu(S) \coloneqq \sum_{s\in S} \mu(s)$ to denote the mass of $S$ with respect to $\mu$. When $\mu(S)>0$, we let $\mu_{S}$ be the distribution of $\mu$ conditioned on $S$. Let $\rho \in \{0, \, \star\}^n$ be a partial assignment corresponding to the sub-cube $C_\rho = \{x \in \ZO^n: \rho_i = 0 \implies x_i = 0 \quad \forall i \in [n]\}$. We use~$\mu_\rho$ to denote $\mu_{C_\rho}$.
\subsubsection{Random partial fixings}
Let $\mu$ be a product distribution over $\ZO^n$. We say that $\mu$ is \emph{$0$-biased} if $\Pr_{\bmx\sim \mu}[\bmx_i=0]\ge 1/2$ for every $i\in [n]$. For the rest of the paper, we will assume without loss of generality that any encountered input distribution is $0$-biased. Indeed, should $\mu$ not be $0$-biased, we can apply the following iterative transformation. Let $f_0\coloneqq f$ and $\mu_0\coloneqq \mu$. For every $i\in [n]$, if $\Pr_{\bmx\sim \mu}[\bmx_i=1]\le 1/2$ -- the coordinate is already biased in the right direction -- we simply leave $f_i\coloneqq f_{i-1}$ and $\mu_i\coloneqq\mu_{i-1}$. Otherwise, let:
\begin{align*}
f_i(x_1,\ldots,x_{i-1},x_i,x_{i+1},\ldots,x_n)\coloneqq f_{i-1}(x_1,\ldots,x_{i-1},1-x_i,x_{i+1},\ldots,x_n);\\
\mu_i(x_1,\ldots,x_{i-1},x_i,x_{i+1},\ldots,x_n)\coloneqq \mu_{i-1}(x_1,\ldots,x_{i-1},1-x_i,x_{i+1},\ldots,x_n).  
\end{align*}
Observe that $\mu_n$ is $0$-biased and $\Dbar_\epsilon(f_n^k,\mu_n^k)=\Dbar_\epsilon(f^k,\mu^k)$ for every $\epsilon \geq 0$ and $k \geq 1$. Now that we are certain that $\mu$ is $0$-biased, let $\delta_i\coloneqq 2\Pr_{\bmx\sim \mu}[\bmx_i=1]\in [0,1]$. We define next the \emph{random partial fixing} distribution with respect to $\mu$. The intuition comes from the observation that each bit of $\mu$ can be written as a convex combination of the fixed bit `$0$' and a uniform bit.
\begin{definition}[Random Partial Fixing]
The random partial fixing with respect to $\mu$, denoted $\cR_\mu$, is a distribution of partial assignments $\bmrho\in \{0,\star\}^n$ sampled as follows: For each $i\in [n]$, we set independently
\begin{equation*}
\bmrho_i= \begin{cases} 0 &\quad\textup{w.p. $1 - \delta_i$}\\ \star &\quad\textup{w.p. $\delta_i$} \end{cases}.
\end{equation*}
\end{definition}
Observe that the following alternative two-step process is equivalent to sampling an input directly from $\mu$. First, sample $\bmrho\sim \cR_\mu$ and then sample and return $\bmx\sim \cU_{\bmrho}$.
\subsubsection{The new measure}
Given a parity decision tree $T$ and a partial assignment $\rho$ over the input string, let $T_\rho$ denote the pruned $T$ by
\begin{enumerate}[noitemsep]
    \item fixing all the variables in the support of $\rho$,
    \item removing redundant queries (those can be written as a linear combination of previous queries). 
\end{enumerate}
For randomised parity decision tree $\cT$, we define $\cT_\rho$ as the distribution of $\bm{T}_\rho$, where $\bm{T}\sim \cT$.
\begin{definition}
For every randomised parity decision tree $\cT$ and product distribution $\mu$, define the skew average cost $\sq(\cT, \mu) \coloneqq \expected_{\bmrho\sim \cR_\mu}[\qbar(\cT_{\bmrho},\cU_{\bmrho})]$. Let $f\colon \ZO^n \to \ZO$ be a function. For $\epsilon \geq 0$, we define the skew measure $\S_\epsilon(f)$ with:
\[
\S_\epsilon(f,\mu)\coloneqq \min\nolimits_{\cT}\left\{\sq(\cT,\mu)\mid \err_{f}(\cT,\mu)\le \epsilon \right\}.
\]
\end{definition}
\begin{claim}\label{claim:S_leq_avgD}
For any $f\colon\ZO^n\to \ZO$, product distribution $\mu$, and $\epsilon \geq 0$, $\Dbar_\epsilon(f,\mu) \geq \S_\epsilon(f,\mu)$. Furthermore, equality holds if $\mu = \cU$.
\end{claim}
\begin{proof}
The claim is immediate as $\sq(\cT,\mu) \le \qbar(\cT,\mu)$ for every randomised parity tree $\cT$ and product distribution $\mu$.
\end{proof}
\subsection{Proof plan}
The proofs of \Cref{theorem:direct_sum_for_avgD,theorem:direct_sum_for_constant_bounded_avgD} are carried out in two steps. First, we prove a perfect direct sum for the skew measure in \Cref{section:Direct_Sum_for_S}.
\begin{restatable}{theorem}{DirectProductForS}\label{theorem:ds_for_M}
We have $\S_\epsilon(f^k,\mu^k)\ge k\cdot \S_\epsilon(f,\mu)$ for any function $f$, product $\mu$ and $\epsilon \geq 0$. 
\end{restatable}
As a second step, we demonstrate in \Cref{section:Compression} that $\Dbar_\epsilon(f, \mu) \approx \S_\epsilon(f, \mu)$. We first prove a lossless conversion for product distribution which are \emph{constant-bounded}. We then extend this to general product distributions for which we lose a $\log(n)$-factor. Let us recall here that the $\log n$ loss for general (unbounded) product distribution is inherent to the skew measure. Indeed, we show in \cref{section:separation} the existence of some $f$ and $\mu$ for which $\Dbar_{1/3}(f, \mu) = \Theta(\log n)$ but $\S_0(f, \mu) = \Theta(1)$. 
\begin{restatable}{theorem}{BoundAvgDArbitraryMu}\label{theorem:bound_avgD_arbitrary_mu}
For any $f\colon\ZO^n\to \ZO$, product distribution $\mu$, $\gamma\in (0,1/n)$, we have
\[
\Dbar_{\epsilon+\gamma}(f,\mu) \leq O\big(\log(n/\gamma)\big) \cdot (\S_{\epsilon}(f,\mu)+1) \quad \forall \epsilon \geq 0.
\]
\end{restatable}
\begin{restatable}{theorem}{BoundAvgBoundedMu}\label{theorem:bound_avgD_bounded_mu}
For any $f \colon \ZO^n \to \ZO$ and $\lambda$-bounded product distribution $\mu$, we have
\[
\Dbar_{\epsilon}(f,\mu) \leq O(1/\lambda) \cdot \S_{\epsilon}(f,\mu) \quad \forall \epsilon \geq 0.
\]
\end{restatable} 
Combining the results above it is now straightforward to conclude and prove \cref{theorem:direct_sum_for_avgD,theorem:direct_sum_for_constant_bounded_avgD}. For instance, the proof of the former goes as follows.
\begin{proof}[Proof of \cref{theorem:direct_sum_for_avgD}]
\begin{align*}
\Dbar_\epsilon(f^k,\mu^k) &\ge \S_\epsilon(f^k,\mu^k) &\text{(\Cref{claim:S_leq_avgD})}\\
&\ge k\cdot\S_\epsilon(f,\mu)&\text{(\Cref{theorem:ds_for_M})}\\
&\geq \Omega\big(k/\log(n/\gamma)\big) \cdot (\Dbar_{\epsilon+\gamma}(f,\mu)-C\cdot \log(n/\gamma)). &\text{(\Cref{theorem:bound_avgD_arbitrary_mu})} &\qedhere
\end{align*} 
\end{proof}
\subsection{Some notation}
Let us finish this section by defining some notations which will be useful for the rest of the paper. Let $T$ be a parity decision tree on input $\ZO^n$. We define $\cN(T)$ as the set of nodes of $T$ and $\cL(T)$ as the set of leaves of $T$. For each node $v\in \cN(T)$, we define the following: (items marked with $\bm{*}$ are only defined for non-leaf nodes)
\begin{itemize}[noitemsep]
\item $\PATH(v)$: the set of nodes on the root-to-$v$ path (including the root, excluding $v$)
\item $d(v)\coloneqq |\PATH(v)|$: the depth of $v$
\item[$\bm{*}$] $Q^v\in \ZO^n$: the query made at node $v$
\item[$\bm{*}$] $\child(v,b)$ the child of $v$ corresponding to the query outcome $\dotprod{x}{Q^v} = b$, where $b\in \ZO$
\item $Q^{\prec v}$: an $n\times d(v)$ boolean matrix with column vectors $\{Q^u\}_{u\in \PATH(v)}$
\item[$\bm{*}$] $Q^{\preceq v}\coloneqq [Q^{\prec v}\; Q^v]$ of dimension $n\times (d(v)+1)$.
\item $b^{\prec v}\in \ZO^{d(v)}$: the labels on the root-to-$v$ path
\end{itemize}
For every boolean matrix $A\in \ZO^{n\times m}$, we use $\rank(A)$ to denote the rank of $A$ (understood as a matrix over $\mathbb{F}_2)$ and let $\col(A)\subseteq \ZO^n$ be the column space of $A$. For every $S\subseteq [n]$, let $A_S\in \ZO^{|S|\times m}$ stand for the sub-matrix of $A$ consisting of row with indices in S. For every $x,y\in \ZO^n$ and $S\subseteq [n]$, we denote $\dotprod{x_S}{y_S} = \sum_{i\in S}x_iy_i$ by $\dotprod{x}{y}_S$.

Let $\mu$ and $\nu$ be two distributions over $S$. We use $\dTV(\mu,\nu)\coloneqq \sup_{S'\subseteq S} |\mu(S')-\nu(S')|$ to denote the total variation distance between $\mu$ and $\nu$ and write $\mu \equiv \nu$ if $\dTV(\mu,\nu)=0$.
\section{Direct sum for \texorpdfstring{$\D^\times$}{D product} part II: direct sum for \texorpdfstring{$\S$}{S}}\label{section:Direct_Sum_for_S}
In this section, we prove a perfect direct sum for $\S$ (restated below). A direct consequence of this fact is a perfect direct sum for distributional parity query complexity under the uniform distribution.
\DirectProductForS*
\begin{corollary}\label{theorem:uniform}
We have $\Dbar_\epsilon(f^k, \cU^k) \geq k \cdot \Dbar_\epsilon(f, \cU)$ for any function $f$ and $\epsilon \geq 0$.
\end{corollary}
\begin{proof}
Combine \cref{claim:S_leq_avgD} with \cref{theorem:ds_for_M}.
\end{proof}
To prove \cref{theorem:ds_for_M}, our overall strategy is to take a tree achieving $\S_\epsilon(f^k, \mu^k)$ and extract a tree computing a single copy of $f$ under $\mu$ to within error $\epsilon$ while having cost bounded by $\S_\epsilon(f^k, \mu^k)/k$. To do so, we employ the extraction strategy hinted at in \cref{section:technical_overview}. The extractor works as long as the input distributions are uniform, which is the case after the random partial fixing step of $\S$.
\subsection{Extracting a single instance under uniform distributions}
Let $T$ be a deterministic parity tree taking inputs $x \in  \mathcal{X} \coloneqq \ZO^{m_1} \times \dots \times \ZO^{m_k}$ and returning labels in $\ZO^k$.
We assume without loss of generality that the queries along any root-to-leaf path are linearly independent. Let $L(\ell)\in \ZO^k$ be the label associated with the leaf $\ell \in \mathcal{L}(T)$. For $i \in [k]$, we define the linear subspace $W_i \subseteq \mathcal{X}$ of query vectors that are zero everywhere except for copy $i$:
\begin{equation*}
W_i\coloneqq \big\lbrace w \in \mathcal{X}: w^j = 0^{m_j} \iff j \neq i\big\rbrace.
\end{equation*}
We say a node $v \in \cN(T)$ is \emph{critical} with respect to $i$ if $\col(Q^{\prec v})\cap W_{i}\ne \col(Q^{\preceq v})\cap W_{i}$ and denote the set of critical indices at node $v$ with $I_v \coloneqq \{i \in [k]: v \text{ is critical w.r.t. } i\}$. Finally, we let $d_i(v) \coloneqq \sum_{u \in \PATH(v)} \indicator{i\in I^u}$ be the relative depth of $v$ with respect to instance $i$ and highlight that $d_i(v)=\dim(\col(Q^{\prec v})\cap W_i)$. The algorithm $\Ext{T}{i}$ which extracts a tree for the $i$-th instance out of $T$ is described in \cref{Alg:Extractor_Uni}. 
\begin{algorithm}
\begin{algorithmic}[1]
\Input{$y\in \ZO^{m_i}$}
\Output{$a\in \ZO$}
\State{Initialize $v\leftarrow \text{ root of } \bm{T}$}
\While{$v$ is not a leaf}
\If{$i\in I^v$}
\State{Let $w$ be any vector in $(\col(Q^{\preceq v})\setminus \col(Q^{\prec v}))\cap W_{i}$, query $\dotprod{y}{w}$}
\State{Compute $b^v\coloneqq \dotprod{y}{Q^v}$ from $b^{\prec v}$ and $\dotprod{y}{w}$}\label{line:compute_xor}
\State{Move $v\leftarrow \child(v,b^v)$}
\Else
\State{Sample $\bm{\xi}\sim \ber(1/2)$}
\State{Move $v\leftarrow \child(v,\bm{\xi})$}
\EndIf
\EndWhile
\State{\textbf{return} $L_i(v)$}
\end{algorithmic}
\caption{$\Ext{T}{i}$}
\label{Alg:Extractor_Uni}
\end{algorithm}
Observe that it is indeed possible to compute the value of $\dotprod{y}{Q^v}$ from $b^{\prec v}$ and $\dotprod{y}{w}$ on line~\ref{line:compute_xor}: Since $w \notin \col(Q^{\prec v})$, we have $\rank([Q^{\prec v} \:\: w])=\rank(Q^{\prec v})+1$.
On the other hand, as $w\in \col(Q^{\preceq v})$, we have $\rank([Q^{\preceq v} \:\: w])=\rank(Q^{\preceq v})=\rank(Q^{\prec v})+1$. Thus $Q^v\in \col([Q^{\prec v}\:\:w])$, which means that $Q^v$ can be written as a linear combination of the columns of $[Q^{\prec v}\:\:w]$: $Q^v=Q^{u_1} + \cdots +  Q^{u_t} + w$ where $u_1,\ldots,u_t$ are some ancestors of $v$. This in turn implies that $\dotprod{y}{Q^v} =\sum_{i \in [t]} \dotprod{y}{Q^{u_i}} + \dotprod{y}{w}$.\\

We stress that although $T$ is a deterministic tree, $\Ext{T}{i}$ is a randomized decision tree with internal randomness inherited from the bits $\bm{\xi}$. Our main technical claim is that for any fixed $y \in \ZO^{m_i}$, the algorithm $\Ext{T}{i}$ perfectly simulates a run of $T$ when the input is on a random input $\bmx = (\bmx^1, \ldots, \bmx^{i-1}, y, \bmx^{i+1}, \ldots, \bmx^k)$ and $\bmx^j \sim \cU(\ZO^{m_j})$. In a nutshell, the randomness of the other $k-1$ instances can be substituted with the internal randomness $\bm{\xi}$. To make this precise, we let $X_v = \{x \in \mathcal{X}: x^TQ^{\prec v} = b^{\prec v}\}$ be the set of inputs leading to the node $v \in \cN(T)$.
\begin{claim}\label{claim:uniform_simulation}
For any $y \in \ZO^{m_i}$, $\Pr_{\bm{\xi}}[\textup{$\Ext{T}{i}$ reaches node $v$ in its execution on $y$}] = \Pr_{\bmx}[\bmx \in X^v]$.
\end{claim}
\begin{proof}
Let us fix $i \coloneqq 1$ and $d \coloneqq d(v)$ for simplicity. We establish and alternative description of~$X^v$ that puts pure constraints on instance 1 first. Pick $t \coloneqq d_1(v)$ independent vectors $Q_1, \ldots , Q_t \in \col(Q^{\prec v})\cap W_1$ and extend them arbitrarily to a basis $\{Q_j\}_{j \in [d]}$ of $Q^{\prec v}$. As each vector of this basis can be expressed as a linear combination of $\{Q^u\}_{u\in \PATH(v)}$, it is possible to apply those linear combinations to $b^{\prec v}$ and obtain values $\{b_j\}_{j \in [d]}$ such that $X_v = \{x \in \mathcal{X}\mid \forall j\in [d]: \dotprod{x}{Q_j} = b_j \}$. The set $Y^v \subseteq \ZO^{m_1}$ of inputs that can reach node $v$ in a run of $\Ext{T}{1}$ thus corresponds to
\begin{equation*}
Y^v \coloneqq \big\lbrace y \in \ZO^{m_1}\mid  \forall j\in [t]: \dotprod{y}{Q_j^1}=b_j  \big\rbrace.
\end{equation*}
If $y \notin Y^v$, the statement follows directly as both probabilities are zero. However, if $y \in Y^v$,
\begin{equation*}
\Pr\nolimits_{\bm{\xi}}[\textup{$\Ext{T}{1}$ reaches node $v$ in its execution}] = 2^{-d+t}.
\end{equation*}
This is so because a node $v$ can only be reached by having the ``right" $d-t$ coin tosses of $\bm{\xi}$ (provided that $y \in Y^v$). Thus, it remains to show that $\Pr_{\bmx}[\bmx \in X^v] = 2^{-d + t}$ if $y \in Y^v$.

Let $m = \sum_{i \in [k]} m_i$ and $S = \{m_1 + 1, \dots, m\}$ be the indices of the bits of every copy but the first one. Fix the $m \times (d-t)$ boolean matrix $A =[Q_{t+1}\cdots Q_{d}]$ and observe that $\rank(A) = d - t$ by construction. We show that $\rank(A_S) = d - t$ too. If $\rank(A_S) < \rank(A)$, we can find a non-empty set $J\subseteq \{t+1, \ldots, d\}$ such that $\sum_{j\in J} (Q_j)_{S}=0$. This implies that $Q' \coloneqq \sum_{j\in J} Q_j \in W_i \cap \col(Q^{\prec v})$. But $Q'$ is linearly independent of $\{Q_1, \ldots, Q_t\}$ -- this contradicts $\dim(\col(Q^{\prec v}) \cap W_i) = t$. Therefore, if $y \in Y^v$, we use this observation to conclude:
\begin{align*}
\Pr\nolimits_{\bmx}[\bmx \in X^v] &= \Pr\nolimits_{\bmx}[\forall j \in [d]: \dotprod{\bmx}{Q_j} = b_j ]\\
&= \Pr\nolimits_{\bmx}[\bmx^TA = (b_j)_{t+1\le j\le d}]\\
&= \Pr\nolimits_{\bm{z} \coloneqq (\bmx^2, \dots, \bmx^k)}\left[\bm{z}^T A_S = (b_j + \dotprod{y}{Q_j^1})_{t+1\le j\le d} \right]\\
&= 2^{-\rank(A_S)}\\
&=2^{-d+t}. \qedhere
\end{align*}
\end{proof}
\subsection{Proof of \Cref{theorem:ds_for_M}}
We are now ready to show \Cref{theorem:ds_for_M}. Let $\cT$ be a randomised parity decision tree which witnesses $C \coloneqq \S_\epsilon(f^k, \mu^k)$. For each $i \in [k]$, define the randomized decision tree $\mathcal{T}_i\colon \ZO^n \to \ZO$ with:
\begin{enumerate}[noitemsep]
\item Sample $\bm{T} \sim \mathcal{T}$.
\item Sample $\bmrho^1, \, \ldots, \, \bmrho^{i-1}, \, \bmrho^{i+1}, \, \ldots, \bmrho^k \sim \cR_\mu$.
\item Let $\widetilde{\bmrho} \coloneqq (\bmrho^1, \, \ldots, \, \bmrho^{i-1}, \, \star^n, \, \bmrho^{i+1}, \, \ldots, \bmrho^k)$.
\item Return $\Ext{T_{\widetilde{\bmrho}}}{i}$.
\end{enumerate}
We show in \cref{lemma:error_Ti} that $\err_{f}(\cT_{i},\mu) \leq \epsilon$ simultaneously for all $i \in [k]$. On the other hand, we show in \cref{lemma:depth_Ti} that $\sum_{i \in [k]} \sq(\cT_{i},\mu) \leq C$. By an averaging argument, this shows the existence of a copy $i^\star \in [k]$ with cost $\leq C/k$ and therefore $\S_\epsilon(f, \mu) \leq C/k$. The remainder of this section is devoted to proving both claims.
\begin{lemma}\label{lemma:error_Ti}
For every $i\in [k]$, $\err_{f}(\cT_{i},\mu)\le \err_{f^k}(\cT,\mu^k)$.
\end{lemma}
\begin{proof}
It is enough to prove the statement assuming $\mathcal{T}$ is a deterministic parity tree $T$ and $i = 1$. Let $\cR$ be the distribution of $\widetilde{\bmrho}$ in the step 3 of generating $\cT_1$. Fix some $\rho \in \supp(\cR)$ and note that $\rho^1 = \star^n$. We also define $\cU^{-1} \coloneqq \cU_{\rho^2} \times \cdots \times \cU_{\rho^k}$. Using \cref{claim:uniform_simulation} on a leaf $\ell \in \cL(T_\rho)$ yields:
\begin{align*}
\Pr_{\bm{y}, \bm{\xi}}[\textup{$\Ext{T_\rho}{1}$ reaches $\ell$ on $\bm{y}$} \,\wedge\, L_1(\ell) \neq f(\bm{y})] &= \expectedsub{\bm{y}}{\Pr_{\bm{\xi}}[\textup{$\Ext{T_\rho}{1}$ reaches $\ell$ on $\bm{y}$}] \cdot  \indicator{L_1(\ell) \neq f(\bm{y})}}\\
&= \expectedsub{\bm{y}}{\Pr_{\bmx^{-1} \sim \cU^{-1}}[(\bm{y}, \bmx^{-1}) \in X^\ell] \cdot  \indicator{L_1(\ell) \neq f(\bm{y})}}\\
&= \Pr_{\bmx \sim \mu \times \cU^{-1}}[\bmx \in X^\ell \, \wedge \, L_1(\ell) \neq f(\bmx^1)].
\end{align*}
Thus:
\begin{align*}
\err_{f}(\cT_1, \mu) &= \expected_{\widetilde{\bmrho} \sim \cR}\big[\Pr\nolimits_{\bm{y} \sim \mu,\, \bm{\xi}}[\Ext{T_{\widetilde{\bmrho}}}{1}(\bm{y}) \neq f(\bm{y})]\big]\\
&= \expected_{\widetilde{\bmrho}}\left[\sum\nolimits_{\ell \in \cL(T_{\widetilde{\bmrho}})}\Pr\nolimits_{\bmx \sim \mu \times \cU^{-1}}\big[\bmx \in X^\ell \, \wedge \, L_1(\ell) \neq f(\bmx^1)\big]\right]\\
&\leq \expected_{\widetilde{\bmrho}}\left[\sum\nolimits_{\ell \in \cL(T_{\widetilde{\bmrho}})}\Pr\nolimits_{\bmx \sim \mu \times \cU^{-1}}\big[\bmx \in X^\ell \, \wedge \, L(\ell) \neq f(\bmx)\big]\right]\\
&= \expected_{\widetilde{\bmrho}}\left[\err_{f^k}(T_{\widetilde{\bmrho}}, \mu \times \cU^{-1})\right].
\end{align*}
Observe now that for any $x \in \supp(\mu \times \cU^{-1})$, we have $T_{\widetilde{\bmrho}}(x) = T(x)$. Using the definition of $\cR_\mu$ thus yields:
\begin{equation*}
\err_{f}(\cT_1,\mu) \leq \expected_{\widetilde{\bmrho} \sim \mathcal{R}}\big[\err_{f^k}(T_{\bmrho^{-1}}, \mu \times \cU^{-1})\big] = \err_{f^k}(T, \mu^k). \qedhere
\end{equation*}
\end{proof}
\begin{lemma}\label{lemma:depth_Ti}
$\sum_{i\in [k]} \sq(\cT_{i},\mu)\le \sq(\cT,\mu^k)$.
\end{lemma}
\begin{proof}
It is sufficient to prove this for the case where $\mathcal{T}$ is a deterministic tree $T$. We have:
\begin{align*}
\sum\nolimits_{i \in [k]} \sq(\cT_i,\mu) =& \sum\nolimits_{i \in [k]} \expected_{\bmrho^{i}\sim \cR_\mu}\left[\qbar((\cT_{i})_{\bmrho^{i}},\cU_{\bmrho^{i}})\right]\\
=& \sum\nolimits_{i \in [k]} \expected_{\substack{\bmrho^{i}\sim \cR_\mu\\ \widetilde{\bmrho} \sim \cR}}\left[\qbar\left(\left(\Ext{T_{\widetilde{\bmrho}}}{i}\right)_{\bmrho^{i}},\cU_{\bmrho^{i}}\right)\right]\\
=&\sum\nolimits_{i\in[k]} \expected_{\bmrho \sim \cR_{\mu}^k}\left[\qbar\left(\Ext{T_{\bmrho}}{i},\cU_{\bmrho^{i}}\right)\right]\\
=&\expected_{\bmrho\sim \cR_{\mu}^k}\left[\sum\nolimits_{i\in[k]} \qbar\left(\Ext{T_{\bmrho}}{i},\cU_{\bmrho^{i}}\right)\right].
\end{align*}
where the third equality is due to the fact that the operations of applying $\mathsf{Ext}$ and fixing variables are commutable. Let $\rho \in (\{0, \, \star\}^n)^k$ be a partial fixing and $\ell \in \cL(T_\rho)$. The probability that node $\ell$ is visited during the process $\Ext{T_\rho}{i}$ when the input is $\bmx^i \sim \cU_{\rho^i}$ is $2^{-d(\ell)}$. Observe that $\Ext{T_\rho}{i}$ only makes $d_i(\ell)$ queries to $\bmx^1$ to reach $\ell$. As such, we have:
\begin{align*}
\sum\nolimits_{i\in[k]} \qbar\left(\Ext{T_{\bmrho}}{i},\cU_{\bmrho^{i}}\right) &=  \sum_{i \in [k]} \sum_{\ell \in \cL(T')} 2^{-d(\ell)} d_i(\ell)\\
& \leq \sum_{\ell \in \cL(T')} 2^{-d(\ell)} d(\ell)\\
&= \qbar(T_\rho, \cU_\rho).
\end{align*}
The inequality is due to the fact that $\sum_{i \in [k]} d_i(v) \leq d(v)$. This is because $\dim(W_i\cap W_j)=0$ for each $i \neq j$ and so
\begin{equation*}
\sum\nolimits_{i\in [k]} d_i(v)=\sum\nolimits_{i\in [k]} \dim(\col(Q^{\prec v})\cap W_i)\le \dim(\col(Q^{\prec v}))=d(v). 
\end{equation*}
To conclude, we have
\[
\sum\nolimits_{i \in [k]} \sq(\cT_i,\mu)=\expected_{\bmrho\sim \cR_{\mu}^k}\left[\sum\nolimits_{i\in[k]} \qbar\left(\Ext{T_{\bmrho}}{i},\cU_{\bmrho^{i}}\right)\right]\le \expectedsub{\bmrho\sim \cR_\mu^k}{\qbar(T_{\bmrho}, \cU_{\bmrho})}=\sq(T,\mu^k). \qedhere
\]
\end{proof}
\section{Direct sum for \texorpdfstring{$\D^\times$}{D product} part III: from \texorpdfstring{$\S$}{S} to \texorpdfstring{$\D^\times$}{D product}}\label{section:Compression}
In this section, we show how to convert parity tree of the $\S_\epsilon$ model to the more common $\Dbar_\epsilon$ model and prove \cref{theorem:bound_avgD_arbitrary_mu,theorem:bound_avgD_bounded_mu}. Let us fix for this section a boolean function $f\colon \ZO^n \to \ZO$ together with some $0$-biased product distribution $\mu$ over $\ZO^n$. Let $T$ be a deterministic parity tree trying to solve $f$ against $\mu$. We begin by establishing an alternative view of the quantity $\sq(T, \mu)$. For any fixed $x \in \ZO^n$, define the product distribution $R_\mu^x$ over $\ZS^n$ with:
\begin{equation}\label{equation:distribution_R_x}
\Pr_{\bmrho\sim \cR_\mu^x}[\bmrho_i=\star] = \begin{cases}
    \delta_i/(2-\delta_i) &\quad\textup{if $x_i=0$}\\
    1 &\quad\textup{if $x_i=1$}
\end{cases} \quad\quad \text{where} \quad \delta_i\coloneqq 2 \cdot \Pr_{\bmx\sim \mu}[\bmx_i=1]\in [0,1].
\end{equation}
Sampling $\bmrho \sim R_\mu$, $\bmx\sim \cU_{\bmrho}$ and completing $\bmx_j = 0$ for all $\bmrho = 0$ is equivalent to first sampling $\bmx \sim \mu$ and then some $\bmrho \sim R_\mu^{\bmx}$. One can therefore see the process of $\sq(T, \mu)$ as follows:
\begin{enumerate}
\item Sample $\bmx \sim \mu$, $\bmrho \sim R_\mu^{\bmx}$.
\item Run $T$ on $\bmx$.
\item Every time $T$ attempts to make a query, check if $\bmrho$ simplifies the query: $\bmrho_i = 0 \implies \bmx_i = 0$.
\end{enumerate}
We describe this alternative view in detail in \cref{algorithm:altervative_view}. With this new interpretation, we can recast the quantity $\sq(T, \mu)$ with
\begin{equation}\label{equation:alternative_view_sq}
\sq(T, \mu) = \expected\nolimits_{\bmx \sim \mu, \bmrho \sim R_\mu^{\bmx}}[\text{Number of times 
\cref{alg_line:query_in_alternative} is executed in \cref{algorithm:altervative_view}}].
\end{equation}
The idea to convert $\S_\epsilon$ algorithms to $\Dbar_\epsilon$ ones is to simulate the process of \cref{algorithm:altervative_view} by maintaining an incomplete but consistent view $p \in \ZSQ^n$ of $\rho$. Initially, $p = \qmark^n$ -- i.e. nothing is known about $\rho$ -- and we gradually update $p$ based on the queries we get. For instance, if $x_i = 1$, then \cref{equation:distribution_R_x} asserts $\rho_i = \star$. This scheme helps to relate the cost of the converted $\Dbar_\epsilon$ algorithm with $\sq(T, \mu)$. The description of the converted algorithm is given in \cref{algorithm:compressed_bounded_bias}.
\begin{algorithm}[p]
\begin{algorithmic}[1]
\Input{$x\in \ZO^n$, $\rho \in \ZS^n$}
\Output{$a \in \ZO$}
\State{$v \leftarrow \text{ root of } T$}
\While{$v$ is not a leaf} \label{alg_line:sq_start}
\If{$\rank(Q^{\preceq v}_{S(\rho,\star)}) = \rank(Q^{\prec v}_{S(\rho,\star)}) + 1$} \label{alg_line:sq_end}
\State{Query $b^v \gets \dotprod{x}{Q^v}$} \label{alg_line:query_in_alternative}
\Else
\State{Infer $b^v \gets \dotprod{x}{Q^v}$ from the fact that $(Q^{\prec v})^Tx=b^{\prec v}$ and $x_j=0$ for all $\rho_j=0$}
\EndIf
\State{Move $v\leftarrow \child(v,b^v)$.}
\EndWhile
\State{\textbf{return} $L(v)$}
\end{algorithmic}
\caption{an alternative view of $\sq(T, \mu)$}
\label{algorithm:altervative_view}
\end{algorithm}
\begin{algorithm}[p]
\begin{algorithmic}[1]
\Input{$x \in \ZO^n$}
\Output{$a \in \ZO$}
\State{$v\leftarrow \text{ root of } T$}
\State{$p\leftarrow \qmark^n$}
\While{$v$ is not a leaf} \label{alg_line:start_of_while}
\State{$D^{v,p} \leftarrow \{j \in [n] : p_j = \qmark \text{ and } \rank(Q^{\preceq v}_{S(p,\star) + j}) = \rank(Q^{\preceq v}_{S(p,\star)})+1\}$}
\If{$D^{v,p}=\emptyset$}
\State{Infer $b^v\gets \dotprod{x}{Q^v}$ from the fact that $(Q^{\prec v})^Tx=b^{\prec v}$ and $x_j=0$ for all $p_j=0$}
\Else
\For{$j \in D^{v,p}$}\label{line:start_for}
\State{Query $x_j$}
\State{Sample $\bm{\eta}\sim \ber(\delta_j/(2-\delta_j))$}
\If{$x_j=1$ or $\bm{\eta}=1$}
\State{$p_j\leftarrow \star$}
\State{\textbf{break}}
\EndIf
\State{$p_j\leftarrow 0$}
\EndFor\label{line:end_for}
\State{Query $b^v \leftarrow \dotprod{x}{Q^v} $}
\EndIf
\State{Move $v\leftarrow \child(v,b^v)$} \label{alg_line:after_update}
\EndWhile
\State{\textbf{return} $L(v)$}
\end{algorithmic}
\caption{converts an algorithm $T$ for $\S_\epsilon$ to $\Dbar_\epsilon$}
\label{algorithm:compressed_bounded_bias}
\end{algorithm}
\begin{definition}
Let $p \in \ZSQ^n$ be a fixing. The following are subsets of indices:
\[
    S^p_\star = \{j \in [n]: p_j = \star\} \quad
    S^p_0 = \{j \in [n]: p_j = 0\} \quad
    S^p_{\qmark} = \{j \in [n]: p_j = \qmark\} \quad
    S^p_{\neq 0} = S^p_\star \cup S^p_{\qmark}
\]
We also write $S(p,\star)$ to mean $S^p_\star$ and likewise for other sets.
\end{definition}
Let $P^v \subseteq \ZSQ^n$ be the set of all possible $p$ that could be at the start of an iteration of \cref{algorithm:compressed_bounded_bias} at node $v$. We now prove an invariant of \cref{algorithm:compressed_bounded_bias} and then its correctness. 
\begin{lemma}\label{lemma:rank_invariant}
For any state $v \in \cN(T)$ and $p \in P^v$ that \cref{algorithm:compressed_bounded_bias} could be in at the start of a while iteration (\cref{alg_line:start_of_while}), it holds that:
\[
\rank\Big(Q^{\prec v}_{S(p, \neq 0)}\Big) = \rank\Big(Q^{\prec v}_{S(p, \star)}\Big) = |S(p, \star)|.
\]
\end{lemma}
\begin{proof}
We prove the claim by induction on $T$. The statement is true when $v$ is the root because both $Q^{\prec v}$ and $S^p_\star$ are empty. Let us now assume that the statement is true for some $v$ and $p \in P^v$ and prove that the invariant carries over to the next iteration regardless of the query outcomes and the randomness $\bm{\eta}$ of the process. If $p'$ is the updated value of $p$ at \cref{alg_line:after_update}, this amounts to showing that $\rank(Q^{\preceq v}_{S(p', \ne 0)}) = \rank(Q^{\preceq v}_{S(p',\star)}) =|S(p', \star)|$. We consider three cases.\\

\noindent\textbf{Case $D^{v, p} = \emptyset$:} Then, there is no update for $p$ and $p' = p$. Since $\rank(Q^{\preceq v}_{S(p,\star)}) = \rank(Q^{\preceq v}_{S(p,\star) + j})$ for all $j \in S(p,\ne 0)$, we have $\rank(Q^{\preceq v}_{S(p, \ne 0)}) = \rank(Q^{\preceq v}_{S(p,\star)}) = |S(p,\star)|$, as desired.\\

\noindent\textbf{Case $D^{v, p} \neq \emptyset$ and $p'_j=0$ for all $j\in D^{v, p}$:} Then, $S^{p'}_\star = S^p_\star$ and $S(p', \neq 0) = S(p, \neq 0) \setminus D^{v, p}$. By definition of $D^{v,p}$, we still have $\rank(Q^{\preceq v}_{S(p,\star)})=\rank(Q^{\preceq v}_{S(p,\star) + j})$ for all $j\in S(p', \neq 0)$, so $\rank(Q^{\preceq v}_{S(p',\ne 0)}) = \rank(Q^{\preceq v}_{S(p',\star)}) = |S(p', \star)|$.\\

\noindent\textbf{Case $D^{v,p} \neq \emptyset$ and $p'_j=\star$ for some $j\in D^{v,p}$:} Then $S^{p'}_\star = S_\star^p + j$ and it must hold that $\rank(Q^{\preceq v}_{S(p',\star)})=|S(p',\star)|$. On the other hand,
\[
\rank\Big(Q^{\preceq v}_{S(p',\ne 0)}\Big) \leq \rank\Big(Q^{\prec v}_{S(p, \neq 0)}\Big) + 1 = |S(p, \star)| + 1 = |S(p', \star)|.
\]
Where the inequality follows from the fact that $S(p', \ne 0) \subseteq S(p, \neq 0)$. Finally, this implies $\rank(Q^{\preceq v}_{S(p', \ne 0)}) = \rank(Q^{\preceq v}_{S(p', \star)}) = |S(p,\star)|$.
\end{proof}
\begin{lemma}\label{lemma:correctness_converted_alg}
For any $x \in \ZO^n$, $\Pr_{\bm{\eta}}[\textup{\cref{algorithm:compressed_bounded_bias} outputs $1$}] = \indicator{T(x) = 1}$.
\end{lemma}
\begin{proof}
It is not hard to see that if \cref{algorithm:compressed_bounded_bias} gets the correct value of $\dotprod{x}{Q^v}$ at each iteration of the while loop, it perfectly simulates $T$. Thus, it suffices to show that whenever $D^{v,p}=\emptyset$, the algorithm can compute the value of $\dotprod{x}{Q^v}$ from the previous query outcomes. \cref{lemma:rank_invariant} and its proof implies that if $D^{v,p} = \emptyset$, then $\rank(Q^{\preceq v}_{S(p, \ne 0)})=\rank(Q^{\prec v}_{S(p,\ne 0)}) = |S(p,\star)|$. Thus $Q^v_{S(p, \ne 0)}$ can be written as a linear combination of column vectors of $Q^{\prec v}_{S(p, \ne 0)}$. Namely, $Q^v_{S(p, \ne 0)} = \sum_{j \in [t]} Q^{v_j}_{S(p, \ne 0)}$, where $v_1,\ldots,v_t$ are some ancestors of $v$. On the other hand, we know that $x_j=p_j=0$ for all $j\in S^p_0$. Consequently, we have
\[
\dotprod{x}{Q^v} = \dotprod{x}{Q^v}_{S(p, \ne 0)} = \sum\nolimits_{j \in [t]} \dotprod{x}{Q^{v_j}}_{S(p, \neq 0)} = \sum\nolimits_{j \in [t]} b^{v_j}.
\]
Thus, \cref{algorithm:compressed_bounded_bias} follows the same path of vertices as $T$, irrespective of the randomness $\bm{\eta}$. Consequently, its outputs correspond to the one of $T$.
\end{proof}
We now turn our attention to the efficiency of \cref{algorithm:compressed_bounded_bias}. We shall start with the special case of $\mu$ being a constant-bounded distribution. In this particular case, we obtain a lossless conversion.
We then turn our attention to general product distributions, for which \cref{algorithm:compressed_bounded_bias} suffers a $\log(n)$ factor. This loss factor is inherent to reducing $\S_\epsilon$ to $\Dbar$ as \cref{section:separation} shows.
\subsection{Conversion for constant-bounded distribution}\label{subsection:constant_bounded_distribution}
We now prove a strong efficiency result for \cref{algorithm:compressed_bounded_bias} in the special case where $\mu$ is $\lambda$-bounded (see \cref{definition:lambda_bounded}). A proof of our goal (\cref{theorem:bound_avgD_bounded_mu}) then follows easily.
\begin{lemma}\label{lemma:efficiency_bounded_distribution}
We have $\qbar(\textup{\cref{algorithm:compressed_bounded_bias} on $T$}, \mu) \leq (2/\lambda) \cdot \sq(T, \mu)$.
\end{lemma}
Before proving this, we need an alternative view of the randomness used in the for-loop of \cref{algorithm:compressed_bounded_bias} (line~\ref{line:start_for}~to~\ref{line:end_for}). At the start of the process, a random partial fixing $\bmrho \sim \cR_\mu^x$ is generated. The algorithm is then deterministic: whenever some $x_j$ is queried in the for-loop, this is replaced by a query to $\bmrho_j$. The algorithm updates $p_j$ with $\bmrho_j$ and exits the loop if $\bmrho_j = \star$. This process is given in detail in \cref{algorithm:compressed_bounded_bias_alternative}.
\begin{algorithm}[ht]
\begin{algorithmic}[1]
\Input{$x \in \ZO^n$}
\Output{$a \in \ZO$}
\State{$v\leftarrow \text{ root of } T$}
\State{$p\leftarrow \qmark^n$}
\State{Sample $\bmrho \sim R_\mu^x$}
\While{$v$ is not a leaf}
\State{$D^{v,p} \leftarrow \{j \in [n] : p_j = \qmark \land \rank(Q^{\preceq v}_{S(p,\star) + j}) = \rank(Q^{\preceq v}_{S(p,\star)}+1\}$} \label{alg_line:alt_start}
\If{$D^{v,p}=\emptyset$}
\State{Infer $b^v\gets \dotprod{x}{Q^v}$ from the fact that $(Q^{\prec v})^Tx=b^{\prec v}$ and $x_j=0$ for all $p_j=0$}
\Else
\For{$j \in D^{v,p}$}
\State{$p_j \leftarrow \bmrho_j$}\label{alg_line:alternative_query_I}
\If{$p_j = \star$}
\State{\textbf{break}}
\EndIf
\EndFor
\State{Query $b^v \leftarrow \dotprod{x}{Q^v}$}\label{alg_line:alternative_query_II}
\EndIf \label{alg_line:alt_end}
\State{Move $v\leftarrow \child(v,b^v)$.}
\EndWhile
\State{\textbf{return} $L(v)$}
\end{algorithmic}
\caption{an alternative view of \cref{algorithm:compressed_bounded_bias} where the randomness is fixed at the start}
\label{algorithm:compressed_bounded_bias_alternative}
\end{algorithm}
Note that as $R^x_\mu$ is a product distribution, one can actually implement \cref{algorithm:compressed_bounded_bias_alternative} without querying all of $x$ at the start. Indeed, it is enough to query $x_j$ whenever one needs the value of $\bmrho_j$, similarly to \cref{algorithm:compressed_bounded_bias}. This implies that both processes are equivalent.

Suppose one runs \cref{algorithm:compressed_bounded_bias_alternative} on $\bmx \sim \mu$ and $\bmrho \sim R^{\bmx}_\mu$. Fix some state $(v, p)$ the algorithm could be in at the start of the while loop (\cref{alg_line:alt_start}). We let $\cX^{v, p}$ be the distribution of $\bmx$ conditioned on reaching state $(v, p)$. Furthermore, for a fixed $x \in \ZO^n$ and $(v, p)$ reachable with $x$ we let $\cR^{v, p, x}$ be the marginal distribution of $\bmrho$ conditioned on reaching state $(v, p)$ and $\bmx = x$. We now develop explicit formulations for those distributions.
\paragraph{Explicit definition of $\cX^{v, p}$:} Let $\widehat{\cX}^{v, p}$ be the distribution over $\ZO^n$ defined as follows:
\begin{enumerate}[noitemsep, nosep]
\item For all $j \in S^p_0$, fix $\bmx_j = 0$. 
\item For all $j \in S^p_{\qmark}$, sample $\bmx_j \sim \ber(\delta_j/2)$.
\item Determine $\{\bmx_j \colon j\in S^p_\star\}$ by solving $\left\{\dotprod{x}{Q^u}_{S(p,\star)} = \dotprod{x}{Q^u}_{S(p,\neq \star)} +  b^u\right\}_{u\in \PATH(v)}$
\end{enumerate}
\paragraph{Explicit definition of $\cR^{v, p, x}$:} Let $\widehat{\cR}^{p, x}$ be the product distribution over $\ZS^n$ defined as follows:
\begin{enumerate}[noitemsep, nosep]
\item For all $j \in S^p_{\qmark}$ such that $x_j = 0$, let $\bmrho_j = *$ with probability $\delta_j/(2-\delta_j)$ and $\bm\rho_j = 0$ else.
\item For all $j \in S^p_{\qmark}$ such that $x_j = 1$, fix $\bmrho_j = \star$.
\item For all $j \in S(p, \ne \qmark)$, fix $\bmrho_j = p_j$.
\end{enumerate}
\begin{restatable}{claim}{ClaimAlternativeDistribution}\label{claim:alternative_distributions}
For every reachable state $(v,p)$ and $x \in \supp(\cX^{v, p})$ in \cref{algorithm:compressed_bounded_bias_alternative}, we have
\begin{enumerate}[noitemsep]
    \item $\cR^{v,p,x} \equiv \widehat{\cR}^{p,x}$;
    \item $\cX^{v,p} \equiv \widehat{\cX}^{v,p}$.
\end{enumerate}
\end{restatable}
We delay the proof of this technical lemma to \cref{section:appendix_ommited_proofs_compression}. We can now prove the efficiency of our algorithm for $\lambda$-bounded distributions.
\begin{proof}[Proof of \cref{lemma:efficiency_bounded_distribution}]
To relate \cref{algorithm:altervative_view} with \cref{algorithm:compressed_bounded_bias_alternative}, it is helpful to insert the book-keeping of $p$ in \cref{algorithm:altervative_view} (lines \ref{alg_line:alt_start} to \ref{alg_line:alt_end}, without \ref{alg_line:alternative_query_I}) in between \cref{alg_line:sq_start,alg_line:alt_end} of \cref{algorithm:altervative_view}. This doesn't change the number of queries or guarantees of \cref{algorithm:altervative_view} but now both processes share the same state space over $(v, p)$. For $x \in \ZO^n$ and $\rho \in \ZS^n$, define $A(x, \rho)$ and $B(x, \rho)$ as the number of queries each process makes:
\begin{align*}
A(x, \rho) &\coloneqq \text{number of times \cref{alg_line:query_in_alternative} is executed in \cref{algorithm:altervative_view} on input $(x, \rho)$};\\
B(x, \rho) &\coloneqq \text{number of times \cref{alg_line:alternative_query_I,alg_line:alternative_query_II} are executed in \cref{algorithm:compressed_bounded_bias_alternative} on input $(x, \rho)$}.
\end{align*}
Using \cref{equation:alternative_view_sq}, it is thus enough to prove that $\expectedsub{\bmx, \bmrho}{A(\bmx, \bmrho)} \geq \Omega(\lambda) \cdot \expectedsub{\bmx, \bmrho}{B(\bmx, \bmrho)}$ when $\bmx \sim \mu$ and $\bmrho \sim R_\mu^{\bmx}$. We have:
\[
\expected_{\bmx, \bmrho}[A(\bmx, \bmrho)] = \sum_{(v, p)} \Pr_{\bmx, \bmrho}[\text{state $(v, p)$ is reached}] \cdot \Pr\nolimits_{\substack{\bmx\sim \cX^{v,p}\\\bmrho\sim \cR^{v,p,\bmx}}}\Big[\rank(Q^{\preceq v}_{S(\bmrho, \neq 0)}) = \rank(Q^{\prec v}_{S(\bmrho, \neq 0)}) + 1\Big].
\]
As both algorithms follow the same path in the state space, this expectation can be computed with respect to the code of \cref{algorithm:compressed_bounded_bias_alternative}. Fix some state $(v, p)$ and observe that if there exists some $j\in D^{v,p}$ such that $\bmrho_j = \star$, then by \cref{lemma:rank_invariant},
\[
\rank(Q^{\preceq v}_{S({\bmrho},\ne) }) = \rank(Q^{\preceq v}_{S(p,\star) + j}) = \rank(Q^{\prec v}_{S(p,\star)}) + 1 = \rank(Q^{\prec v}_{S(\bmrho, \ne 0)}) + 1.
\]
Therefore, for $\bmx \sim \cX^{v,p}$ and $\bmrho \sim \cR^{v,p,\bmx}$, we have
\begin{align*}
\Pr\nolimits_{\bmx, \bmrho}\Big[\rank(Q^{\preceq v}_{S(\bmrho,\ne 0)}) = \rank(Q^{\prec v}_{S(\bmrho,\ne 0)}) + 1 \Big] &\geq \Pr\nolimits_{\bmx, \bmrho}\left[\exists j \in D^{v, p}: \bmrho_j = \star \right]\\
&= 1 - \Pr\nolimits_{\bmx, \bmrho}[\forall j \in D^{v, p}: \bmrho_j = \bmx_j = 0].
\end{align*}
The last equality is due to the fact that for all $j \in D^{v, p}$, if $\bmrho_j=0$ then $\bmx_j=0$. Let $D \coloneqq D^{v, p}$. We can now substitute $\widehat{\cX}^{v,p}$ for $\cX^{v,p}$ and $\widehat{\cR}^{p, \bmx}$ for $\cR^{v, p, \bmx}$ using \cref{claim:alternative_distributions}:
\begin{align*}
\Pr_{\bmx, \bmrho}[\forall j \in D \colon \bmrho_j = \star \land \bmx_j = 0] &= \Pr_{\bmx, \bmrho}[\forall j \in D \colon \bmx_j = 0] \cdot \Pr_{\bmx, \bmrho}[\forall j \in D \colon \bmrho_j = \star \mid \forall j \in D \colon \bmx_j = 0]\\
&= \prod\nolimits_{j\in D} (1-\delta_j/2) \cdot \prod\nolimits_{j\in D} \frac{2-2\delta_j}{2-\delta_j}\\
&= \prod\nolimits_{j \in D} (1 - \delta_j)\\
&\leq (1 - \lambda)^{|D|}.
\end{align*}
Thus, if $\bmx \sim \mu$ and $\bmrho \sim R_\mu^{\bmx}$, we have
\[
\expected_{\bmx, \bmrho}[A(\bmx, \bmrho)] \geq \sum\nolimits_{(v, p)} \Pr\nolimits_{\bmx, \bmrho}[\text{state $(v, p)$ is reached}] \cdot \Big(1 - (1 - \lambda)^{|D^{v, p}|}\Big).
\]
We now bound the expected number of queries made by $\cT$. When $D^{v,p}=\emptyset$, $\cT$ skips making a query at $v$. On the other hand, when $D^{v, p} \neq \emptyset$, the algorithm goes over $j \in D^{v,p}$ and stops making queries as soon as it hits some $\rho_j = \star$. This probability is independent for each $j \in D^{v, p}$ and can be computed explicitly using \cref{claim:alternative_distributions}. For $\bmx \sim \cX^{v,p}$ and $\bmrho \sim \cR^{v,p,x}$:
\[
\Pr_{\bmx, \bmrho}[\bmrho_j = *] = \Pr_{\bmx}[\bmx_j = 0] \cdot \Pr_{\bmx, \bmrho}[\rho_j = \star \mid \bmx_j = 0] + \Pr_{\bmx}[\bmx_j = 1] \cdot \Pr_{\bmx}[\bmrho_j = \star \mid \bmx_j = 1] = \delta_j \geq \lambda.
\]
Therefore, if $\bmx \sim \mu$ and $\bmrho \sim R_\mu^{\bmx}$,
\begin{align*}
\expected_{\bmx, \bmrho}[B(\bmx, \bmrho)] &\leq \sum\nolimits_{(v, p)} \Pr_{\bmx, \bmrho}[\text{state $(v, p)$ is reached}] \cdot \left(\indicator{D^{v,p} \neq \emptyset} + \sum\nolimits_{j=0}^{|D^{v,p}|-1} (1-\lambda)^j \right)\\
&\leq \sum\nolimits_{(v,p) }\Pr_{\bmx, \bmrho}[\text{state $(v, p)$ is reached}] \cdot \left(\indicator{D^{v,p} \neq \emptyset} + \Big(1-(1-\lambda)^{|D^{v,p}|}\Big)/\lambda\right)\\
&\leq \sum\nolimits_{(v,p) }\Pr_{\bmx, \bmrho}[\text{state $(v, p)$ is reached}] \cdot \left(2/\lambda\right) \cdot \Big(1 - (1 - \lambda)^{|D^{v, p}|}\Big). \qedhere
\end{align*}
\end{proof}
With this in hand, we can now prove \cref{theorem:bound_avgD_bounded_mu}, which we restate below for convenience.
\BoundAvgBoundedMu*
\begin{proof}
Let $\cT$ be a randomised parity tree such that $\sq(\cT, \mu) = \S_\epsilon(f,\mu)$ and $\err_f(\cT,\mu)\le \epsilon$. Define $\cT'$ to be the randomised algorithm obtained by sampling $\bm{T} \sim \cT$ and returning \cref{algorithm:compressed_bounded_bias} applied to $\bm{T}$. Using \cref{lemma:correctness_converted_alg}, we immediately obtain that $\err(\cT', \mu) \leq \epsilon$. On the other hand:
\[
\qbar(\cT', \mu) = \expected_{\bm{T}}\big[\qbar(\text{\cref{algorithm:compressed_bounded_bias} on $\bm{T}$}, \mu)\big] \leq (2/\lambda) \cdot \expected_{\bm{T}}[\sq(\bm{T}, \mu)] = (2/\lambda) \cdot \S_\epsilon(f,\mu).
\]
Thus, $\Dbar_\epsilon(f, \mu) \leq O(1/\lambda) \cdot \S_\epsilon(f, \mu)$, as desired.
\end{proof}
\subsection{Conversion for general product distribution}
\Cref{algorithm:compressed_bounded_bias} is not efficient for arbitrary product distribution since queries can be very biased so that $\prod_{j\in D^{v, p}} (1-\delta_j) = 1-o(1)$. In such cases, we cannot even afford to pay one query as the corresponding expected increment for $\sq$ is $o(1)$.

To overcome this obstacle, we introduce the following idea. Run the algorithm as if every query $x_j$ returned $0$, i.e. assuming $x_j = \bmrho_j = 0$ for all $j \in S(p, \qmark)$ (this is likely to happen for very biased distributions). This generates a list of indices for which we assume $x_j = 0$. Upon reaching a leaf, we check efficiently whether one of those $x_j$ is actually $1$. If no such $j$ exists, we're done -- at the cost of \emph{no real} queries! On the other hand, if a $1$ is found, we backtrack to this state and restart the procedure. Since we've found $x_j = 1$, it must be that $\rho_j = \star$ and the $\S_\epsilon$ algorithm has to pay one query there.

The process $\buildlist$ that ``runs assuming $x_j = 0$'' and produces a list of indices to check is described in \cref{algorithm:build_list}. Then, the updated algorithm for converting an $\S_\epsilon$ algorithm to a $\Dbar_\epsilon$ one is formulated in \cref{algorithm:convert_general_distribution}.
\begin{algorithm}[p]
\begin{algorithmic}[1]
\Input{$x\in \ZO^n$}
\Output{$a\in \ZO$}
\State{Initialize $v\leftarrow \text{ root of } T$, $p\leftarrow \qmark^n$}
\While{$v$ is not a leaf}
\State{$(J,\ell)\gets \buildlist(v,p)$}
\State{Find the first element $i^* \in J$ with $x_{i^*} = 1$ or set $i^* = \bot$ if none exists}\label{alg_line:FFO}
\State{$\textsc{found} \gets 0$}
\For{$j\in J$}\label{alg_line:start_for2}
\State{Sample $\bm{\eta}\sim \ber(\delta_j/(2-\delta_j))$}
\If{$j = i^*$ or $\bm{\eta} = 1$}
\State{$p_j \gets \star$} \label{alg_line:star_line}
\State{$u \gets w_j$}
\State{$\textsc{found} \gets 1$}
\State{\textbf{break}}
\EndIf
\State{$p_j\gets 0$}
\EndFor\label{alg_line:end_for2}
\If{$\textsc{found} = 1$}
\State{Query $\dotprod{x}{Q^u}$ and set $b^u$ as the outcome}
\State{Move $v\leftarrow \child(u,b^u)$}
\Else
\State{Update $v\leftarrow \ell$}
\EndIf
\EndWhile
\State{\textbf{return} $L(v)$}
\end{algorithmic}
\caption{converts an algorithm for $\S_\epsilon$ to $\Dbar_\epsilon$ for general product distributions}
\label{algorithm:convert_general_distribution}
\end{algorithm}
\begin{algorithm}[p]\label{algorithm:build_list}
\begin{algorithmic}[1]
\Input{$v\in \cN(T)$, $p\in \ZSQ^n$}
\Output{a list of indices $J$ and a leaf $\ell$}
\State{Initialize $J \gets []$, $u \gets v$, $p' \gets p$}
\While{$u$ is not a leaf}
\State{$D^{v, p'} \leftarrow \{ j \in [n]\colon p'_j = \qmark \land \rank(Q^{\preceq u}_{S(p', \star) + j}) = \rank(Q^{\preceq u}_{S(p', \star)}) + 1\}$}
\For{$j\in D^{v,p}$} \Comment{in arbitrary order}
\State{$p'_j \leftarrow 0$}
\State{$w_j \leftarrow u$}
\State{$J \leftarrow [J,j]$}
\EndFor
\State{Infer $b^u\leftarrow \dotprod{x}{Q^u}$ assuming $(Q^{\prec v})^Tx=b^{\prec v}$ and $x_j = 0$ for all $j \in S(p', 0)$}
\State{Move $u\leftarrow \child(u,b^u)$}
\EndWhile
\State{\textbf{return} $(J,u)$}
\end{algorithmic}
\caption{the subroutine $\buildlist$}
\end{algorithm}
\paragraph{How to run \cref{alg_line:FFO}?} This problem can be formulated as follows. Let $\FFO_n: \ZO^n \to [n] \cup \bot$ be the search problem that asks for the index of the first (running from left to right) '1' in $x$ or $\bot$ if $x = 0^n$. Even though a simple adversary argument shows that one cannot perfectly compute $\FFO_n$ by making $ < n$ parity queries, a folklore result \cite{Feige1990, Nisan1993, Harms2024}, proves that there is a randomised protocol making $O(\log n)$ queries that computes $\FFO_n$ with some small error.
\begin{restatable}{lemma}{restatableAmplificationTrick}
\label{lemma:amplification_trick}
For any $\alpha > 0$, $\R_\alpha(\textup{\FFO}_n) \leq O\big(\log n+\log(1/\alpha)\big)$.
\end{restatable}
\begin{proof}
This folklore fact is discussed for the parity context in \cref{section:appendix_fact_on_trees}.
\end{proof}
We let $\cT'_\gamma$ be the parity tree obtained by running \cref{algorithm:convert_general_distribution} with error parameter $\alpha \coloneqq \gamma/n$ on \cref{alg_line:FFO}. Given two indices $i,j\in J$, we say $i\prec_J j$ if $i$ appears strictly earlier than $j$ in $J$, and $i\preceq_J j$ if $i\prec_J j$ or $i=j$. Fix any $x \in \supp(\cX^{v,p})$. Let $i^*$ denote the first index $i$ in $J$ such that $x_i=1$ and suppose that $i^*$ is added to $J$ when $u=u^*$. Observe that if such $i^*$ exists, $x_j=0$ for all $j\prec_J i^*$. As a consequence, we know that $u^*$ must be reached. Moreover, we can immediately get the values of $\bmrho_j$ by flipping biased coins for all $j\preceq_J i^*$. Therefore, given $i^*$, one can perfectly simulate \Cref{algorithm:compressed_bounded_bias} by going over $J$ and updating $p$, until finding the first index $j^*\preceq_{J} i^*$ such that $\bmrho_{j^*}=\star$. We are now ready to prove the correctness and efficiency of $\cT'_\gamma$.
\begin{lemma}\label{lemma:correctness_converted_general_distribution}

For any fixed $x \in \ZO^n$, $\Pr[\cT'_\gamma(x) = 1] \in \indicator{T(x) = 1} \pm \gamma$.
\end{lemma}
\begin{proof}
The randomness of $\cT'_\gamma$ stems from $\bm{\eta}$ and the randomness involved in running the $\FFO$ algorithm at \cref{alg_line:FFO}. To analyse the latter, observe that \cref{alg_line:FFO} is called at most $n$ times and each call fails with probability at most $\alpha = \gamma/n$, hence:
\[
\dTV(\cT'_0(x),\cT'_\gamma(x)) \leq \Pr[\text{at least one oracle call at line~\ref{alg_line:FFO} gives a wrong index}] \leq n \cdot (\gamma/n) = \gamma.
\]
If no call fails the discussion above implies that \cref{algorithm:convert_general_distribution} behaves identically to the earlier \cref{algorithm:compressed_bounded_bias}. Hence, correctness of the former (\cref{lemma:correctness_converted_alg}) implies $\Pr[\cT'_0(x) = 1] = \indicator{T(x) = 1}$.
\end{proof}
\begin{lemma}\label{lemma:efficiency_converted_general_distribution}
We have $\qbar(\cT'_\gamma, \mu) \leq O(\log(n/\gamma)) \cdot (\sq(T,\mu)+1) +\gamma \cdot n$.
\end{lemma}
\begin{proof}
We first prove that the expected number of iterations of the outer while-loop is low assuming that the algorithm always gets the correct index $i^*$ at line~\ref{alg_line:FFO}. Similar to what we did in \cref{subsection:constant_bounded_distribution}, we view the randomness used in the for-loop (line~\ref{alg_line:start_for2} to \ref{alg_line:end_for2}) in \cref{algorithm:convert_general_distribution} as a pre-generated partial assignment $\bmrho\sim \cR_\mu^x$. Note that the bits of $\bmrho$ are independent. If $i^*$ is the first index in $J$ with $x_{i^*} = 1$, we know that $x_{i^*}=1$ and $x_j=0$ for all $j\prec_{J} i^*$. At the same time, $\bmrho_j$ for all $j\preceq_{J} i$ are revealed to the algorithm one by one. As soon as some $\bmrho_j=\star$ is found, the algorithm quits the loop.

For each $x\in \ZO^n$ and $\rho \in \supp(\cR_\mu^x)$, consider running $\cT'_\gamma$ on input $x$ using randomness $\rho$.
Define $K(x,\rho)$ as the number of iterations of the outer while loop when $\cT'_\gamma$ always gets the correct $i^*$ on line~\ref{alg_line:FFO}. Let $p^*$ denote the final state of $p$. Since in each iteration except for the last one, we update some $p_j$ as $\star$, we have $K(x,\rho) \leq |S(p^*, \star)| + 1$. By \cref{lemma:rank_invariant}, we further have
\[
K(x,\rho) \leq \rank\left(Q^{\prec \ell(x)}_{S(p^*, \star)}\right) + 1 = \rank\left(Q^{\prec \ell(x)}_{S(p^*, \neq 0)}\right) + 1,
\]
where $\ell(x)\in \cL(T)$ is the unique leaf at which $T$ terminates given $x$. Since for all $p_j\ne ?$, $p_j=\rho_j$, we have $S^{p^*}_\star\subseteq S^{\rho}_\star\subseteq S^{p^*}_{\ne 0}$, hence $K(x, \rho) \leq \rank(Q^{\prec \ell(x)}_{S(\rho, \star)}) + 1$.
On the other hand, by definition we have 
\[
\sq(T,\mu) = \expected_{\substack{\bmx \sim \mu\\\bmrho \sim \cR_\mu^{\bmx}}}
\Big[\rank\Big(Q^{\prec \ell(\bmx)}_{S(\bmrho, \star)}\Big)\Big] \quad\implies\quad \expected_{\substack{\bmx \sim \mu\\\bmrho \sim \cR_\mu^{\bmx}}}[K(\bmx,\bmrho)] \leq \sq(T,\mu)+1.
\]
\Cref{lemma:amplification_trick} asserts that line~\cref{alg_line:FFO} can be implemented to error $\gamma/n$ using $O(\log(n/\gamma))$ parity queries. Since all those calls are completed successfully with probability $\geq \gamma$, we finally have:
\begin{equation*}
\qbar(\cT'_\gamma, \mu)\leq (1-\gamma) \cdot \expected\nolimits_{\substack{\bmx \sim \mu\\\bmrho \sim \cR_\mu^x}}[K(\bmx, \bmrho)] \cdot O(\log(n/\gamma)) + \gamma \cdot n \leq O(\log(n/\gamma)) \cdot (\sq(T,\mu)+1) +\gamma \cdot n. \qedhere
\end{equation*}
\end{proof}
\BoundAvgDArbitraryMu*
\begin{proof}
Let $\cT$ be a randomised parity decision tree such that $\sq(\cT, \mu) = \S_\epsilon(f,\mu)$ and $\err_f(\cT,\mu)\le \epsilon$. Define $\cT^*$ to be the randomised algorithm obtained by sampling $\bm{T} \sim \cT$ and returning the corresponding $\cT'_\gamma$. Using \cref{lemma:correctness_converted_general_distribution}, we immediately obtain that $\err(\cT^*, \mu) \leq \epsilon + \gamma$. By \cref{lemma:efficiency_converted_general_distribution} and the range of parameters allowed for $\gamma$, we get
\begin{equation*}
\qbar(\cT^*, \mu)  = \expected\nolimits_{\bm{T}}\big[\qbar(\cT'_\gamma, \mu) \big] \leq O(\log(n/\gamma)) \cdot \expected\nolimits_{\bm{T}}[\sq(T,\mu) + 1] = O(\log(n)/\gamma) \cdot (\sq(\cT, \mu) + 1). \qedhere
\end{equation*}
\end{proof}
\section{Separations I: \texorpdfstring{$\disc$}{disc} vs. \texorpdfstring{$\D^\times$}{prod}} \label{section:disc-vs-Dprod}
In this section we prove \cref{lemma:disc-vs-dprod}, restated here for convenience.
\DiscVsDprod*
\begin{proof}
For the first item, we can consider the $n$-bit majority function $f\coloneqq \MAJ_n$. It follows from~\cite[Theorem 1.2]{Braverman15} that~$\D^\times(\MAJ_n) \geq \Omega(n)$ where the hard distribution is uniform. By contrast, it is not hard to see that~$\disc(\MAJ_n) \leq O(\log n)$ (if we query $x_i$ for a random $i\in[n]$, it will have bias $\geq \Omega(1/\sqrt{n})$ toward predicting $\MAJ_n(x)$). We prove the second item by a probabilistic argument. Consider a random function $\bm{f}$, which is set with $\bm{f}(x) \sim \ber(2^{-0.9n})$ independently for each $x \in \ZO^n$. In \cref{claim:random_f_disc}, we show that $\disc(\bm{f})=\Theta(n)$ and in \cref{claim:random_f_dprod} that $\D^\times(\bm{f})=O(1)$ with high probability.
\end{proof}
\begin{claim}\label{claim:random_f_disc}
With probability $1-2^{-2^{\Omega(n)}}$, $\disc(\bm{f})\ge 0.01n$.
\end{claim}
\begin{proof}
For each non-constant function $f:\ZO^n\to \ZO$, we define the ``hard'' distribution $\mu_f$ as
\begin{equation*}
    \mu_f(x) \coloneqq \begin{cases}
        1/(2|f^{-1}(0)|) & \quad\text{if $f(x)=0$}\\
        1/(2|f^{-1}(1)|) & \quad\text{if $f(x)=1$}
    \end{cases}.
\end{equation*}
To prove the claim, it suffices to show $\Pr_{\bm{f}}[\disc(\bm{f},\mu_{\bm{f}})\ge 0.01n]\ge 1-2^{-2^{\Omega(n)}}$. Using \Cref{lemma:disc_characterisation}, this can be further simplified to prove:
\[
\Pr_{\bm{f}}\left[\max\nolimits_{S\in \cO_n}\bias(\bm{f},\mu_{\bm{f}},S)\le 2^{-0.01n-1}\right]\ge 1-2^{-2^{\Omega(n)}}.
\]
To that end, fix any $S\in \cO^n$, note that $|S|=|\ZO\setminus S|=2^{n-1}$ and observe that by a Chernoff bound,
\begin{align*}
\Pr\nolimits_{\bm{f}}\left[|\mu(\bm{f}^{-1}(1)\cap S)-1/4| \ge 2^{-0.02n}]\right] &\leq \Pr\nolimits_{\bm{f}}\left[|\bm{f}^{-1}(1)|<2^{0.1n-1}\right]\\
&\quad\quad +\Pr\nolimits_{\bm{f}}\left[||\bm{f}^{-1}(1)\cap S|-2^{0.1n-1}|>2^{0.07n}\right]\\
&\quad\quad +\Pr\nolimits_{\bm{f}}\left[||\bm{f}^{-1}(1)\setminus S|-2^{0.1n-1}|>2^{0.07n}\right]\\
&\leq 3e^{-2^{0.03n}}.
\end{align*}
Using a similar argument, we can also show $\Pr_{\bm{f}}[|\mu(\bm{f}^{-1}(0)\cap S)-1/4|\ge 2^{-0.02n}]\le 3e^{-2^{0.03n}}$.
By definition, $\bias(\bm{f},\mu_{\bm{f}},S)=|\mu(\bm{f}^{-1}(0)\cap S)-|\mu(\bm{f}^{-1}(1)\cap S)|$, we thus have $\Pr[\bias(\bm{f},\mu_{\bm{f}},S)\ge 2^{-0.01n-1}]\le 6e^{-2^{0.03n}}$. Finally, observe that $|\cO_n| \leq 2^n$ and so using a union bound,
\begin{align*}
\Pr_{\bm{f}}[\disc(\bm{f})\ge 0.01n]&\ge \Pr_{\bm{f}}[\max_{S\in \cO_n}\bias(\bm{f},\mu_{\bm{f}},S)\le 2^{-0.01n-1}]\\
&\ge 1-2^n\Pr[\bias(\bm{f},\mu_{\bm{f}},S)\ge 2^{-0.01n-1}]\\
&\ge 1-2^{-2^{\Omega(n)}}. \qedhere
\end{align*}
\end{proof}
\begin{claim}\label{claim:random_f_dprod}
With probability $1-2^{-\Omega(n)}$, $\D^\times(f)\le 20000$.    
\end{claim}
\begin{proof}
Let $\cD^\times\coloneqq \{\ber(p_1,\ldots,p_n)\mid p_1\ldots,p_n\in [0,1/2]\}$ denote the set of $0$-biased product distributions, where $\ber(p_1,\ldots,p_n)\coloneqq \ber(p_1)\times \cdots \times\ber(p_n)$.
As observed in \Cref{section:main_distributional}, it suffices to show $\Pr_{\bm{f}}[\max_{\mu\in \cD^\times} \D_{1/3}(\bm{f},\mu)\le 20000]\ge 1-2^{-\Omega(n)}$.

As a first attempt, one might want to prove that $\D_{1/3}(f,\mu)=O(1)$ with sufficiently high probability for any fixed $\mu$ and then apply union bound over all $\mu\in \cD^\times$.
However, this cannot be done directly since $\cD^\times$ is infinite.
Luckily, we can circumvent this barrier by discretizing $\cD^\times$. Let us define $\discreteD\coloneqq \{\ber(a_1/10n,\ldots,a_n/10n)\mid a_1,\ldots,a_n\in \{0,\ldots,5n\}\}$. For every $\mu=\ber(p_1,\ldots,p_n)\in \discreteD$ and $f:\ZO^\to \ZO$, consider the following two cases:
\begin{itemize}
\item If $\sum_{i} p_i\ge 10$, then $M\coloneqq \max_{x\in \ZO^n} \mu(x)\le e^{-\sum_{i} p_i}\le 1000\sum_{i}p_i$. Observe that 
\[
\Pr_{\bm{f}}\left[\sum\nolimits_{x\in \ZO} f(x)\mu(x)\ge 1/5\right]\le 2^M\cdot (2^{-0.9n})^{M/5}\le 2^{-150\sum_{i}p_in},
\]
thus $\Pr_{\bm{f}}[\D_{1/4}(\bm{f},\mu)=0]\ge 1-2^{-150\sum_{i}p_in}$.
\item Otherwise, we devise the following protocol:
Sort $\mu(x_1)\ge \cdots \ge\mu(x_{2^n})$.
Pick the top $1000$ inputs $X=\{x_1,\ldots,x_{1000}\}$,
then we check if our input $x$ is in $X$.
If yes, we output $f(x)$, otherwise we output $0$.
Formally, we define the function $g:\ZO^n\to \ZO$ where
\begin{equation*}
g(x) \coloneqq \begin{cases} f(x) &\quad \text{if $x\in X$}\\ 0 &\quad \text{if $x \notin X$} \end{cases}.
\end{equation*}
Since testing whether $x=x_i$ can be done with $m$ queries with success probability $1-2^{-m}$, by choosing $m=20$ and running the testing for every $i\in [1000]$, one can show $\cR_{0.01}(g)\le 20000$. It remains to prove that $\Pr_{\bm{f}}[\bm{f}(x)=\bm{g}(x)]\ge 4/5$ with high probability. Observe that for each $x\notin X$, $\mu(x)\le 1/1000$. Therefore:
\[
\Pr_{\bm{f}}\left[\sum\nolimits_{x\notin X}[\mu(x)\bm{f}(x)]\le 1/5 \right]\ge 1-2^{1000}\cdot (2^{-0.9n})^{200}\ge 1-2^{-150n}.
\]
For those $\bm{f}$, we have $\Pr_{\bm{f}}[\bm{f}(x)=\bm{g}(x)]\ge 4/5$, which implies that $\D_{0.22}(\bm{f},\mu)\le 20000$.
\end{itemize}
By union bound over $\mu\in \discreteD$, we can deduce that
\begin{align*}
\Pr_{\bm{f}}\left[\max_{\mu\in \discreteD}\D_{0.22}(\bm{f},\mu)> 20000\right] & \leq \sum\nolimits_{\mu\in \discreteD} \Pr\nolimits_{\bm{f}}[\D_{0.22}(\bm{f},\mu)> 20000] \\
& \leq \sum\nolimits_{a_1=0}^{5n}\cdots \sum\nolimits_{a_n=0}^{5n} \indicator{\sum\nolimits_i a_i\ge 100n}\cdot e^{-150\sum_{i}a_i}\\
&\quad\quad + \sum\nolimits_{a_1=0}^{5n}\cdots \sum\nolimits_{a_n=0}^{5n} \indicator{\sum\nolimits_i a_i< 100n}\cdot 2^{-150n}\\
&\leq \sum\nolimits_{a_1=0}^{5n} \cdots \sum\nolimits_{a_n=0}^{5n} e^{-100(a_1+5)}\cdots e^{-100(a_n+5)}+2^{101n}\cdot 2^{-150n}\\
&\leq \left(\sum\nolimits_{a_1=0}^{5n} e^{-100(a_1+5)}\right)^n+2^{-49n}\\
&\leq 2^{-\Omega(n)}.
\end{align*}
Consider now any product distribution $\mu=\ber(p_1,\ldots,p_n)\in \cD^\times$, define its rounded version $\lceil \mu \rceil$:
\[
\lceil \mu \rceil\coloneqq \left(\ber\left(\frac{\lceil 10n\cdot p_1\rceil}{10n}\right),\, \ldots,\, \ber\left(\frac{\lceil10n\cdot p_n\rceil}{10n}\right)\right)\in \discreteD.
\]
Observe that $\dTV(\mu,\lceil \mu \rceil)\le 1-(1-1/10n)^n\le 1-1/e^{-1/10}<0.1$, thus we have $\err_f(T,\lceil \mu \rceil)\le \err_f(T,\mu)+0.1$ for any parity tree $T$ and $f:\ZO^n \to \ZO$.
Together with the string of inequalities developed above, we conclude that with probability at least $1-2^{-\Omega(n)}$,
\begin{equation*}
\max\nolimits_{\mu\in \cD^{\times}} \D_{1/3}(\bm{f},\mu)\le \max\nolimits_{\mu\in \discreteD} \D_{1/3-0.1}(\bm{f},\mu)\le \max\nolimits_{\mu\in \discreteD} \D_{0.22}(\bm{f},\mu)\le 20000. \qedhere
\end{equation*}
\end{proof}
\section{Separations II: \texorpdfstring{$\S$}{S} vs. \texorpdfstring{$\D^\times$}{prod}} \label{section:separation}
The goal of this section is to provide the following example of a function.
\begin{theorem}\label{theorem:separation_1}
There exists a function $f\colon\ZO^n \to \ZO$ and a product distribution $\mu$ such that $\D^\times(f)=\Theta(\disc(f, \mu)) = \Theta(\log n)$ and $\S_0(f, \mu) = \Theta(1)$.
\end{theorem}
Recall that by \cref{theorem:bound_avgD_arbitrary_mu}, this is the largest possible gap between $\S$ and $\D^\times$. To prove the separation, we use the function $\FPE\colon \ZO^{2n} \to \ZO$ which takes two inputs $x,y \in \ZO^n$ and returns the value $y_i$ associated with the location $i$ of the first `$1$' in $x$. More precisely, we let $\FO(x) \in [n]$ be the location (from left to right) of the first `$1$' in $x$ and $\FO(x) = 1$ if $x = 0^n$ and let $\FPE(x, y) = y_{\FO(x)}$. We choose as hard distribution the product distribution $\mu \coloneqq \mathcal{X} \times \mathcal{Y}$ where for each $i \in [n]$:
\[
\mathcal{X}_i \sim \ber(1/\sqrt{n}) \quad\textup{and}\quad \mathcal{Y}_i \sim \ber(1/2).
\]
Let us note that the choice of $1/\sqrt{n}$ in the distribution $\mathcal{X}$ is arbitrary: any $p = n^a$ for constant $a \in (-1, 0)$ is enough to guarantee that $x \neq 0^n$ with high probability and get the $\Omega(\log n)$ lower bound.
\begin{proof}[Proof of \cref{theorem:separation_1}]
We first prove that $\S_0(\FPE,\mu)=\Theta(1)$.
Consider the following simple brute-force query algorithm $T$ that computes $f$: Query the bits of $x$ one by one from left to right, until finding the first index $i$ such that $x_i=1$. Then query $y_i$ and return $y_i$ if such $i$ exists. Otherwise ($x=0^n)$, simply return $1$.

Observe that $\err_\FPE(T,\mu)=0$.
Thus we only need to show $\sq(f,\mu)=\Theta(1)$.
Let $X_i\coloneqq \{x\mid x_i=1,x_j=0,\forall j<i\}$ denote the set of $x\in \ZO^n$ for which $\FO(x)=i$.
Note that $\ZO^n=X_1\sqcup \cdots \sqcup X_n\sqcup \{0^n\}$ forms a partition of $\ZO^n$.
By the definition of $\mu$, we have $\mu(X_i)=(1-1/\sqrt{n})^{i-1}/\sqrt{n}$.
For all $x\in X_i$, $T$ queries the same set of variables $\{x_1,\ldots,x_{i-1},x_i,y_i\}$ on $x$.
Moreover, sample $\bmrho\sim \cR_\mu^x$ and for each $1\le j<i$, since $x_j=0$, we have that $\Pr[\bmrho_j=\star]=1/(\sqrt{n}-1)$. Therefore,
\[
h(x)\coloneqq \expected\nolimits_{\bmrho\sim \cR^x_\mu}[q(T_{\bmrho},x)]\le \frac{i-1}{\sqrt{n}-1}+2.
\]
We conclude that
\begin{align*}
\sq(T,\mu)&=\expectedsub{\bmx \sim \mu}{h(\bmx)}\\
&\le \sum\nolimits_{i=1}^n \mu(X_i)\cdot \expected\nolimits_{\bmx \sim \mu_{X_i}}[h(\bmx)]+(n+1)\cdot\mu(0^n)\\
&\le \frac{1}{n-\sqrt{n}}\cdot\sum\nolimits_{i=1}^n (i-1)(1-1/\sqrt{n})^{i-1}+n\cdot (1-1/\sqrt{n})^n+2\\
&< \frac{2}{n} \cdot \sum\nolimits_{i=0}^{\infty} i(1-1/\sqrt{n})^i+3\\
&=\Theta(1).
\end{align*}
Let us now turn our attention to $\disc(\FPE,\mu)$.
The lower-bound $\disc(\FPE, \mu) \geq \Omega(\log n)$ is covered in \cref{claim:disc_lb}. The upper bound $\disc(\FPE, \mu) \leq O(\log n)$ is a direct consequence of $\bias(\FPE, \mu, S) \geq n^{-1/2}$ for $S = \{(x, y) \in \ZO^n: y_1 = 1\}$. More interestingly, one can actually show the stronger statement $\D_{1/3}(f, \mu) \leq O(\log n)$. Indeed, $\bmx \sim \mathcal{X}$ has exactly one `$1$' in the first $\lceil\sqrt{n}\rceil$ bits with probability $\geq e^{-1.01} \geq 1/3$ for $n$ large enough. In that case, a simple binary search amongst the first $\lceil\sqrt{n}\rceil$ bits of $x$ using parity queries is enough to find that location and return the corresponding bit of $y$.
\end{proof}
\begin{claim}\label{claim:disc_lb}
$\disc(\textup{\FPE}, \mu) \geq \Omega(\log n)$
\end{claim}
\begin{proof}
Using the characterisation of the bias with codimension-1 subspace \cref{lemma:disc_characterisation}, it is enough to show:
\[
\max\nolimits_{S \in \mathcal{O}^n} \bias(\FPE, \mu, S) \leq n^{-1/3}.
\]
Fix an affine space $S^\star$ of codimension 1 that maximize the above expression, i.e. some $\alpha, \beta \in \ZO^n$ and $\gamma \in \ZO$ such that $S^\star = \{(x, y) \in \ZO^{2n}: \alpha \cdot x + \beta \cdot y = \gamma\}$. To simplify notation, we assume in what follows that $\gamma = 0$ but the proof is similar for the case $\gamma = 1$. Let us partition $S$ in two sets:
\begin{align*}
S^0 &\coloneqq \{(x, y) \in \ZO^{2n}: \alpha \cdot x = 0 \text{ and } \beta \cdot y = 0\};\\
S^1 &\coloneqq \{(x, y) \in \ZO^{2n}: \alpha \cdot x = 1 \text{ and } \beta \cdot y = 1\}.
\end{align*}
We have:
\[
\max\nolimits_{S \in \mathcal{O}^n} \bias(\FPE, \mu, S) = \bias(\FPE, \mu, S^\star) \leq \bias(\FPE, \mu, S^0) + \bias(\FPE, \mu, S^1).
\]
Let us suppose without loss of generality that $\bias(\FPE, \mu, S^0) \geq \bias(\FPE, \mu, S^1)$ so that it is enough to show $\bias(\FPE, \mu, S^0) \leq 2n^{-1/2}$. Note that if $\Pr_{\bmx,\bm{y} \sim \mu}[(\bmx,\bm{y}) \in S^0] = 0$, we're done. If not, we can re-express the bias in the language of probability:
\begin{align*}
\bias(\FPE, \mu, S^0) &= \left\vert \sum\nolimits_{(x,y) \in S^0} (-1)^{\FPE(x)} \mu(x)\right\vert\\
&= \left\vert \sum\nolimits_{b \in \ZO} (-1)^b \cdot \Pr_{\bmx,\bm{y}}\left[\FPE(x) = b \wedge (\bmx,\bm{y}) \in S^0\right]\right\vert\\
&= \Pr_{\bmx,\bm{y}}\left[(\bmx,\bm{y}) \in S^0\right] \cdot \left\vert \sum\nolimits_{b \in \ZO} (-1)^b \cdot \Pr_{\bmx,\bm{y}}\left[\FPE(\bmx) = b \, | \, (\bmx,\bm{y}) \in S^0\right]\right\vert.
\end{align*}
Let us denote the quantity within the absolute value by $\Phi$. Observe that $S^0$ can be conveniently decomposed as $S^0 = S^X \times S^Y$ where $S^X \coloneqq \{x \in \ZO^n: \alpha \cdot x = 0\}$ and $S^Y \coloneqq \{y \in \ZO^n: \beta \cdot y = 0\}$. With this, we have:
\begin{align*}
\Phi &= \sum\nolimits_{b \in \ZO} (-1)^b \cdot \Pr_{\bmx,\bm{y}}\left[\FPE(\bmx,\bm{y}) = b \, | \, (\bmx,\bm{y}) \in S^0\right]\\
&= \sum_{i \in [n]} \sum_{b \in \ZO} (-1)^b \cdot \Pr_{\bmx,\bm{y}}\left[\FO(\bmx) = i \, | \, (\bmx,\bm{y}) \in S^0\right] \Pr_{\bmx,\bm{y}}\left[\FPE(\bmx, \bm{y}) = b \, | \, (\bmx,\bm{y}) \in S^0 \wedge \FO(\bmx) = i\right]\\
&= \sum_{i \in [n]} \Pr_{\bmx \sim \mathcal{X}}\left[\FO(\bmx) = i \, | \, \bmx \in S^X\right] \cdot \sum_{b \in \ZO} (-1)^b \cdot \underbrace{\Pr_{\bm{y} \sim \mathcal{Y}}\left[\bm{y}_i = b \, | \, \bm{y} \in S^Y\right]}_{\coloneqq p_i^b}.
\end{align*}
Recall that $S^Y$ is a codimension-1 space and $\mathcal{Y}$ is the uniform distribution over $\ZO^n$. Thus, if $|\alpha|$ (the number of non-zero entries in $\alpha$) is zero or $\geq 2$, it must be that $p_i^b = 1/2$ for all $i\in [n]$ and $b \in \ZO$. In that case, the claim is proven because $\Phi = 0$ and so $\bias(\FPE, \mu, S^0) = 0$. We can thus assume that $|\alpha| =1$ and fix $i^* \in [n]$ to be the unique coordinate such that $\alpha_{i^*} = 1$. Now, observe that $p_i^b = 1/2$ for all $i \neq i^*$ and $b \in \ZO^n$, $p_{i^*}^0 = 1$ and $p_{i^*}^1 = 0$ so that:
\[
\Phi = \sum_{i \in [n]} \Pr_{\bmx \sim \mathcal{X}}\left[\FO(\bmx) = i \, | \, \bmx \in S^X\right] \cdot (p_i^0 - p_i^1) = \Pr_{\bmx \sim \mathcal{X}}\left[\FO(\bmx) = i^* \, | \, \bmx \in S^X\right].
\]
Finally, we use the fact that the event $\FO(\bmx) = i^*$ with $\bmx \sim \mathcal{X}$ is unlikely to happen if $S^X$ has large mass under $\mathcal{X}$. In any case, the probability is maximized for $i^* = 1$ and hence:
\begin{align*}
\bias(\FPE, \mu, S^0) &= \Pr\nolimits_{\bmx \sim \mathcal{X}}\left[\bmx \in S^X\right] \cdot \Pr\nolimits_{\bm{y} \sim \mathcal{Y}}\left[\bm{y} \in S^Y\right] \cdot |\Phi|\\
&\leq \Pr\nolimits_{\bmx}\left[\FO(\bmx) = i^* \wedge \bmx \in S^X\right]\\
&\leq \Pr\nolimits_{\bmx}\left[\FO(\bmx) = 1\right
].
\end{align*}
The event $\FO(\bmx) = 1$ can happen because $\bmx_1 = 1$ or $\bmx = 0^n$, thus we bound the bias with
\begin{equation*}
\Pr_{\bmx \sim \mathcal{X}}\left[\FO(\bmx) = 1\right] \leq \Pr_{\bmx}\left[\bmx_1 = 1\right] + \Pr_{\bmx}\left[\bmx = 0^n\right] \leq n^{-1/2} + e^{-\sqrt{n}} \leq 2n^{-1/2}.\qedhere
\end{equation*}
\end{proof}
\appendix
\section{Appendix}
\subsection{Direct sums for \texorpdfstring{$\D$}{D}}
\label{app:deterministic-case}
In this appendix, we prove that the best-known direct sum results in the context of deterministic communication complexity can be obtained in the parity decision tree setting. We restate our theorem for convenience below.
\DeterministicDirectSums*
Let us first introduce a couple of definitions. Fix a function $f\colon \ZO^n \to \ZO$. A parity certificate for $f(x)$ is an affine space $S \subseteq \ZO^n$ such that $x \in S$ and for any $x' \in S$, $f(x) = f(x')$. Similarly to the classical case, the parity certificate complexity $\C(f)$ is the smallest codimension of a space that certifies the value $f(x)$ -- where the hardest possible $x \in \ZO^n$ is taken. We also define $\spar(f) \coloneqq \|\hat{f}\|_0=|\{z\mid \hat{f}(z) \ne 0\}|$ for the number of non-zero Fourier coefficients of $f$. To prove \cref{theorem:deterministic_direct_sums}, it is enough to prove a direct sum for parity certificate complexity and employ the following two results:
\begin{enumerate}[noitemsep]
    \item $\C(f) \geq \D(f)^{1/2}$ \cite{ZS10}
    \item $\C(f) \geq \D(f)/\log \spar(f)$ \cite{Tsang2013}
\end{enumerate}
\begin{lemma}\label{lemma:direct_sum_for_C}
For any $f\colon\ZO^n\to \ZO$ and $k \geq 1$, $\C(f^k) \geq k\cdot \C(f)$.
\end{lemma}
\begin{proof}[Proof of \Cref{lemma:direct_sum_for_C}]
Fix an input $x \in \ZO^n$ attaining $d \coloneqq \C(f)$ and suppose towards contradiction that $\C(f^k) < dk$. This implies in particular that there exists an affine space $S \subseteq (\ZO^n)^k$ described by $m<dk$ equations $Q^Tx=b$ (where $Q\in \ZO^{n\times m},b\in \ZO^m$) that certifies the value of the input $y \in (\ZO^n)^k$ which is composed of $k$ copies of $x$. Define $d_i$ for $i \in [k]$ with:
\[
d_i \coloneqq \dim(\col(Q) \cap W_i) \quad W_i \coloneqq \{w \in (\ZO^n)^k: \, w^j = 0^n \iff j \neq i\}.
\]
Observe that $\sum_{i \in [k]} d_i\le m< dk$ and as such there must be some $i^*$ with $d_{i^*} < k$. Fix for simplicity $i^* = 1$. Using Gaussian elimination, one can re-express $S = S_1 \cap S_2$ where
\begin{enumerate}
\item the constraints in $S_1$ are exclusively on bits of the first copy and
\item any constraint in $S_2$ has at least one bit of a copy other than the first.
\end{enumerate}
Since $S_1$ is about the first copy only, it can be identified with a single-copy affine space $S^* \subseteq \ZO^n$ where $\codim(S^*)=d_1<k$ in a natural way. Observe that $x \in S^*$ as $y \in S$. Because the codimension of $S^*$ is strictly less than $k$, there must be some $x' \in S^*$ with $f(x) \neq f(x')$. Note that fixing $x^1 \coloneqq x'$ leaves the system of linear constraints $S_2$ feasible and as such there exists $x^2, \dots, x^k \in \ZO^n$ such that $y' \coloneqq (x', x^2, \dots, x^k) \in S$: a contradiction since $f(y) \neq f(y')$.
\end{proof}
\subsection{Omitted case of \cref{theorem:main_distributional}}\label{subsection:ommited_case}

\begin{lemma}\label{lemma:low_distributional}
If $\D^\times(f) \leq 6C \cdot \log(n)$, we have $\R(f^k) \geq \Omega(k/\log n) \cdot \D^\times(f)$.
\end{lemma}
\begin{proof}
Fix a hard product distribution $\mu$ for $\D^\times(f)$. If $\D_{1/3}(f, \mu) = 0$, the claim follows trivially. Else, we have $\D_{1/3}(f, \mu) > 0$ and using \cref{claim:zero_query} with $\epsilon \coloneqq 1/6$, it must be that $\S_{1/6}(f) \geq 1/6$. Using \cref{claim:S_leq_avgD} and \cref{theorem:ds_for_M}, we thus have:
\begin{equation*}
\R(f^k) \geq \D_{1/6}(f^k,\mu^k) \geq \S_{1/6}(f^k,\mu^k) \geq k\cdot\S_{1/6}(f,\mu) \geq k/6 \geq \Omega(k/\log n) \cdot  \D^\times(f) \qedhere
\end{equation*}
\end{proof}

\begin{claim}\label{claim:zero_query}
For any $f$, product distribution $\mu$ and $\epsilon \geq 0$, we have $\D_{\epsilon+\S_{\epsilon}(f,\mu)}(f,\mu) = 0$.
\end{claim}
\begin{proof}
Fix a deterministic decision tree $T$ and consider the zero-query decision tree $T'$ that comes out of applying \cref{algorithm:zero_query} to $T$. To relate $T$ and $T'$, we go through \cref{algorithm:convert_general_distribution}. Again, let $\mathcal{T}_0$ be the tree obtained by applying \cref{algorithm:convert_general_distribution} to $T$ with error zero on \cref{alg_line:FFO}. We stress that $\mathcal{T}_0$ is a randomised decision tree depending on $\bm{\eta}$. On the other hand, $T'$ can be seen as $\mathcal{T}_0$ with fewer instructions executed. Using \cref{lemma:correctness_converted_general_distribution}, we have:
\begin{align*}
    \Pr\nolimits_{\bmx \sim \mu}[T(\bmx) \neq T'(\bmx)] &= \Pr\nolimits_{\bmx, \bm{\eta}}[\mathcal{T}_0(\bmx) \neq T'(\bmx)]\\
    &\leq \Pr\nolimits_{\bmx, \bm{\eta}}[\text{\Cref{alg_line:star_line} is executed while running $\mathcal{T}_0(\bmx)$}]\\
    &= \Pr\nolimits_{\bmx, \bmrho \sim \mathcal{R}^{\bmx}_\mu}[\text{$T_{\bmrho}(\bmx)$ makes a query}]\\
    &\leq \mathbb{E}_{\bmx, \bmrho}[q(T_{\bmrho}, \bmx)]\\
    &= \sq(T, \mu)
\end{align*}
Now, let $\cT$ be a randomised parity tree such that $\sq(\cT, \mu) = \S_\epsilon(f,\mu)$ and $\err_f(\cT,\mu)\le \epsilon$. Let $\cT'$ be the randomised parity tree obtained as follows:
\begin{enumerate}
    \item Sample $\bm{T} \sim \cT$
    \item Return $\bm{T}'$ obtained by applying $\bm{T}$ to \cref{algorithm:zero_query}.
\end{enumerate}
With the analysis above, we obtain $\Pr_{\bmx,\bm{T}}[\bm{T}'(\bmx) \neq \bm{T}(\bmx)] \leq \expected_{\bm{T}\sim \cT}[\sq(\bm{T}, \mu)]=\sq(\cT,\mu)$. We remark that $\cT'$ makes no queries and has the following error probability:
\[
\err_f(\cT',\mu)\le \err_f(\cT,\mu)+ \Pr\nolimits_{\bmx,\bm{T}}[\bm{T}'(\bmx) \neq \bm{T}(\bmx)]\le \epsilon+\sq(\cT, \mu)=\epsilon+\S_\epsilon(f,\mu).\qedhere
\]
\end{proof}
\begin{algorithm}
\begin{algorithmic}[1]
\Input{$x\in \ZO^n$}
\Output{$a\in \ZO$}
\State{Initialize $v\leftarrow \text{ root of } T$, $p\leftarrow \qmark^n$}
\State{$(J,\ell)\gets \buildlist(v,p)$}
\State{\textbf{return} $L(\ell)$}
\end{algorithmic}
\caption{converts an algorithm for $\S_\epsilon<1$ to a zero-query algorithm}
\label{algorithm:zero_query}
\end{algorithm}
\subsection{Direct sum for distribution-free discrepancy} \label{section:appendix_disc_distribution_free}
\begin{theorem}\label{theorem:DP_for_DistFree_disc}
For every function $f\colon\ZO^n\to \ZO$ and $k \geq 1$,
\[
k\cdot\disc(f) + 1 \geq \disc(f^{\oplus k})\geq k\cdot \big(\disc(f)-1\big).
\]
\end{theorem}
\begin{proof}
The lower bound is a simple consequence of \cref{lemma:direct_sum_disc} by fixing $\mu$ to be a distribution such that $\disc(f) = \disc(f, \mu)$ and observing that $\disc(f^{\oplus k}) \geq \disc(f^{\oplus k}, \mu^k)$. The other direction is more interesting as it says that the hardest distribution for $f^{\oplus k}$ is basically $k$ products of the hardest distribution for a single copy $f$. Let $\dualnorm{f} \coloneqq \min_{\mu}~\infinitynorm{\widehat{F_\mu}}$ where $\mu$ ranges over all distributions. Using \cref{lemma:disc_characterisation}, we obtain the following relation between $\disc(f)$ and $\dualnorm{f}$:
\begin{equation*}
-\log \dualnorm{f}+1 \geq \disc(f) \geq -\log\dualnorm{f}.
\end{equation*}
Therefore, to prove the upper bound, it is enough to show a perfect direct product for $\dualnorm{f}$ and apply it $k$ time. To this end, fix some other function $g: \ZO^n \to \ZO$ and let us show that
\begin{equation*}
\dualnorm{f \oplus g} \ge \dualnorm{f} \cdot \dualnorm{g}.
\end{equation*}
Where we recall that $f \oplus g: \ZO^{2n} \to \ZO$. We can write $\dualnorm{f}$ as the value of the following linear program where the variables describe a distribution $\mu$:
\begin{equation}\label{eqn:primal_LP}
\begin{aligned}
  \text{min.} &\quad c \\
  \text{s.t.} &\quad \Big\vert\sum\nolimits_{x\in \ZO^n} (-1)^{f(x)} \cdot \mu_x \cdot (-1)^{\dotprod{x}{z}}\Big\vert \le c &\quad \forall z\in\ZO^n\\
  &\quad \sum\nolimits_{x\in \ZO^n} \mu_x = 1\\
  &\quad \mu_x\ge 0 &\quad \forall x \in \ZO^n
\end{aligned}
\end{equation}
The dual of \eqref{eqn:primal_LP} is:
\begin{equation}\label{eqn:dual_LP}
\begin{aligned}
  \text{max.} &\quad d \\
  \text{s.t.} &\quad \sum\nolimits_{z\in \ZO^n} (-1)^{f(x)} \cdot \beta_z \cdot (-1)^{\dotprod{x}{z}} \ge d \quad \forall x \in \ZO^n\\
  &\quad \sum\nolimits_{z\in \ZO^n} |\beta_z|=1
\end{aligned}
\end{equation}
Let $(\beta^f,d^f)$ and $(\beta^g,d^g)$ be the optimal feasible solutions to $\eqref{eqn:dual_LP}$ with respect to $f$ and $g$. By the strong duality theorem, it holds that $\dualnorm{f} = d^f$ and $\dualnorm{g} = d^g$. We now extract a feasible solution for \eqref{eqn:dual_LP} with respect to the function $f \oplus g$. Let $\beta \in \ZO^{2n}$ be defined with $\beta_{(z_1, z_2)} = \beta^f_{z_1} \cdot \beta^g_{z_2}$ and observe that $(\beta, d^f \cdot d^g)$ is a feasible solution for the dual of $\dualnorm{f \oplus g}$. By applying the strong duality theorem again, we have $\dualnorm{f \oplus g}  \geq d^f \cdot d^g = \dualnorm{f} \cdot \dualnorm{g}$, as desired.
\end{proof}
\subsection{Some facts about parity decision trees}\label{section:appendix_fact_on_trees}
Yao's minimax principle is a powerful technique to analyse randomised algorithms -- we adapt here the statement to parity trees, but the proof is exactly the same as the original one \cite{Yao1977}.
\begin{lemma}\label{lemma:yao_minimax}
For any $f\colon \ZO^n \to \ZO$ and distribution $\mu$ over $\ZO^n$, $\R_\epsilon(f) \geq \D_\epsilon(f, \mu)$.
\end{lemma}
The following is a folklore fact relating randomised parity tree complexity and discrepancy \cite{Yao1983, Babai1986} which we re-prove in the parity context.
\begin{lemma}\label{lemma:pR_vs_disc}
$\D_\epsilon(f, \mu) \geq \disc(f, \mu) + \log(1 - 2\epsilon)$ for any $\epsilon \in [0, 1/2)$.
\end{lemma}
\begin{proof}
Fix a parity decision tree $T$ of depth $d \coloneqq \D_\epsilon(f, \mu)$ which makes error $\err_f(T,\mu)\le \epsilon$, note that
\begin{align*}
1-2\epsilon &\leq \Pr_{\bmx \sim \mu}[T(\bmx) = f(\bmx)] - \Pr_{\bmx \sim \mu}[T(\bmx) \neq f(\bmx)]\\
&= \sum\nolimits_{S \in \mathcal{L}} \Pr_{\bmx\sim \mu}[T(\bmx) = f(\bmx) \wedge \bmx \in S] - \Pr_{\bmx\sim \mu}[T(\bmx) \neq f(\bmx) \wedge \bmx \in S].
\end{align*}
As $|\mathcal{L}(T)| \leq 2^d$, there exists some $S \in \mathcal{L}(T)$ -- an affine subspace -- with large correlation:
\begin{equation*}
\bias(f, \mu, S) = \left\vert  \Pr_{\bmx \sim \mu}[T(\bmx) = f(\bmx) \wedge \bmx \in S] - \Pr_{\bmx \sim \mu}[T(\bmx) \neq f(\bmx) \wedge \bmx \in S] \right\vert \geq \frac{1 - 2\epsilon}{2^d}. \qedhere
\end{equation*}
\end{proof}
\restatableAmplificationTrick*
\begin{proof}
Let $\NOR_n: \ZO^n \to \ZO$ be the function that checks whether the input is $0^n$ and rejects otherwise. Observe that one iteration of the sumcheck protocol can be performed in one parity query. More precisely for any $x \in \ZO^n$, if $\bm{s} \sim U\left(\ZO^n\right)$ then:
\[
\Pr_{\bm{s}}[\dotprod{x}{\bm{s}} = 1] = \begin{cases} 1/2 &\quad\text{if $x \neq 0^n$}\\ 0 &\quad\text{if $x = 0^n$} \end{cases}.
\]
Performing two random checks independently shows that $\R(\NOR_n, 1/4) \leq O(1)$. It is a folklore result that a (classical) randomised decision tree can solve $\FFO_n$ with probability $\epsilon$ using $O(\log n + \log(1 / \epsilon))$  oracle $\NOR$-queries even if the oracle fails with probability $1/3$ \cite{Feige1990, Nisan1993}. We highlight that this is an improvement over the naive method that boosts the noisy $\NOR$ queries and yields complexity $O(\log(n)^2\log(1/\epsilon))$. Recent work \cite[\defaultS 3]{Harms2024} revisits this trick in depth for communication complexity and their result can be re-interpreted in the context of parity decision trees as follows:
\[
\forall f\colon\, \R(f, \epsilon) \leq O\big(\D^{\smNOR}(f) + \log(1/\epsilon) \big).
\]
Thus, plugging in $f = \FFO$ and noting that $\D^{\smNOR}(\FFO_n) \leq \log n$ (with binary search), we get the desired result.
\end{proof}
\claimpruningPDT*
\begin{proof}
Let $\cT$ be a randomised PDT satisfying that $d \coloneqq \qbar(\cT,\mu)=\Dbar_\epsilon(f,\mu)$ and $\err_f(\cT,\mu)\le \epsilon$.
To prove the lemma, it suffices to construct a deterministic parity tree $T$ of depth $T \leq d/\gamma$ with $\err_f(T,\mu)\le \epsilon+\gamma$. Sample $\bm{T} \sim \cT$. We construct a new tree $\bm{T}'$ by pruning $\bm{T}$ as follows: We remove all the nodes of $\bm{T}$ of depth greater than $d/\delta$. If any node of depth $d/\delta$ becomes a leaf, we label it with an arbitrary bit. Note that $\bm{T}'$ has depth $\leq d/\delta$. Finally, let $\cT'$ denote the distribution over $\bm{T}'$ inherited from $\cT$.

We observe that for each $x\in \ZO^n$, both $\bm{T}(x)=f(x)$ and $\bm{T}'(x)\ne f(x)$ happen only if $q(\bm{T}, x)>d/\gamma$. Moreover, by Markov's inequality,
\begin{equation*}
\Pr\nolimits_{\substack{\bm{T} \sim \cT\\ \bmx \sim \mu}}[q(\bm{T}, \bmx)>d/\gamma] \leq \frac{\qbar(\bm{T}, \mu)}{d/\gamma} = \gamma.
\end{equation*}
Therefore, $\err_f(\cT',\mu)\le \err_f(\cT,\mu)+\gamma\le \epsilon+\gamma$. By an averaging argument, there exists some $T\in \supp(\cT')$ of depth $\leq d/\delta$ that computes $f$ with error $\err_f(T,\mu)\le \epsilon+\gamma$, as desired.
\end{proof}
\subsection{Omitted proofs of \cref{section:Compression}}\label{section:appendix_ommited_proofs_compression}
In this appendix, we prove \cref{claim:alternative_distributions}, an alternative description for the distributions of \cref{subsection:constant_bounded_distribution}. Let $p^1, p^2 \in \ZSQ^n$. We write $p^1 \sim p^2$ if $p^1$ and $p^2$ are consistent over their non$-\qmark$ entries. That is, $p^1 \sim p^2$ if for all $j \in [n]$, if $p^1_j \neq \qmark$ and $p^2_j \neq \qmark$, then $p^1_j = p^2_j$. \Cref{claim:alternative_distributions} follows from \cref{claim:cR_equiv,claim:cX_equiv}.
\begin{claim}\label{claim:cR_equiv}
For every reachable state $(v,p)$, consistent $x \in \ZO^n$ and $\rho \in \ZS^n$, $\cR^{v, p, x} \equiv \widehat{\cR}^{p, x}$.
\end{claim}
\begin{proof}
Upon inspection of $\widehat{\cR}^{v, p}$, it is enough to prove that for all $x \in \ZO^n$:
\[
\Pr_{\bmrho \sim \cR^{v, p, x}}[\bmrho = \rho] = 
\prod_{j \in S^p_{\ne \qmark}} \indicator{\rho_j = p_j} 
\times \prod\nolimits_{\substack{j \in S^p_{\qmark}\\x_j = 1}} \indicator{\rho_j = \star}
\times \prod\nolimits_{\substack{j \in S^p_{\qmark}\\x_j = 0}} \begin{cases}
    \delta_j/(2-\delta_j) &\textup{if $\rho_j = \star$}\\
    1 - \delta_j/(2-\delta_j) & \textup{if $\rho_j = 0$}
\end{cases}.
\]
Fix $x \in \ZO^n$. We prove this by induction on the state space $(v, p)$ consistent with $x$. The entry-point of the state space is $(\textup{root}(T), \qmark^n)$. In this case, the statement holds by definition. Suppose now that the statement is true for state $(v, p)$. Depending on the value of $\rho$, there are several next state $(v', p')$ possible. Observe however that the next vertex of $T$ to be visited does not depend on $\rho$, as it is fixed to be $v' \coloneqq \child(v, \dotprod{x}{Q^v})$. For any fixed $\rho \in \ZS^n$, we have:
\begin{align*}
\Pr\nolimits_{\bmrho \sim \cR^{v', p', x}}[\bmrho = \rho] &= \Pr\nolimits_{\bmx \sim \mu, \bmrho \sim R_\mu^{\bmx}}[\bmrho = \rho \mid \textup{$(v', p')$ is reached and $\bmx = x$}]\\
&= \frac{\Pr_{\bmx, \bmrho}[\textup{$\bmrho = \rho$ and $(v', p')$ is reached and $\bmx = x$}]}{\Pr_{\bmx, \bmrho}[\textup{$(v',p')$ is reached and $\bmx = x$}]}.
\end{align*}
Note that there can be only one state from which $(v', p')$ can be reached, namely $(v, p)$. Indeed, suppose that there is another state $(v, p^\star)$ from which $(v', p')$ can be reached. Then $(v,p)$ and $(v, p^*)$ have a common ancestor $(u, q)$. Since the paths diverged after $(u, q)$, it must be that $p \nsim p^*$ and thus $p^* \nsim p'$: a contradiction. Thus, we have the following equivalence:
\[
\textup{$(v', p')$ is reached} \quad \iff \quad \textup{$(v, p)$ is reached and $\rho \sim p'$}.
\]
Therefore, we have:
\begin{equation}\label{equation:fraction}
\Pr\nolimits_{\bmrho \sim \cR^{v', p', x}}[\bmrho = \rho] = \frac{\Pr_{\bmrho \sim \cR^{v, p, x}}[\bmrho = \rho] \cdot \indicator{\rho \sim p'}}{\Pr_{\bmrho \sim \cR^{v, p, x}}[\bmrho \sim p']}.
\end{equation}
We can now use the inductive hypothesis on $(v, p)$. Since $\rho \sim p'$ implies $\rho \sim p$, the numerator of~\cref{equation:fraction} simplifies to:
\[
 \prod\nolimits_{j \in S^{p'}_{\ne \qmark}} \indicator{\rho_j = p_j} 
\times \prod\nolimits_{\substack{j \in S^p_{\qmark}\\x_j = 1}} \indicator{\rho_j = \star}
\times \prod\nolimits_{\substack{j \in S^p_{\qmark}\\x_j = 0}} \begin{cases}
    \delta_j/(2-\delta_j) &\textup{if $\rho_j = \star$}\\
    1 - \delta_j/(2-\delta_j) & \textup{if $\rho_j = 0$}
\end{cases}.
\]
Let $\Delta = S^p_{\qmark} \setminus S^{p'}_{\qmark}$ and observe that the denominator of \cref{equation:fraction} is equal to:
\begin{equation*}
\prod\nolimits_{\substack{j \in \Delta\\x_j = 1}} \indicator{\rho_j = \star}
\times \prod\nolimits_{\substack{j \in \Delta\\x_j = 0}} \begin{cases}
\delta_j/(2-\delta_j) &\textup{if $\rho_j = \star$}\\
1 - \delta_j/(2-\delta_j) & \textup{if $\rho_j = 0$}
\end{cases}. \qedhere
\end{equation*}
\end{proof}
\newpage
\begin{claim}\label{claim:cX_equiv}
For every reachable state $(v,p)$ and $x \in \ZO^n$, $\cX^{v, p} \equiv \widehat{\cX}^{v, p}$.
\end{claim}
\begin{proof}
Fix some $(v, p)$ and $x \in \ZO^n$. Upon inspection of $\widehat{X}^{v, p}$, it is enough to prove that
\[
\Pr\nolimits_{\bmx \sim \cX^{v, p}}[\bmx = x] = M(x, v, p) \cdot \prod\nolimits_{j \in S^p_?} 1 - \delta_j/2 - x_j \cdot (1-\delta_j),
\]
where $M(x, v, p)$ is an indicator set to 1 if and only if for all $j \in [n]$, $p_j = 0$ implies $x_j = 0$ and $\dotprod{x}{Q^u} = b^u$ for all $u \in \PATH(v)$. By Baye's rule we have:
\begin{align*}
\Pr_{\bmx \sim \cX^{v,p}}[\bmx = x] &=\Pr\nolimits_{\substack{\bmx\sim \mu\\ \bmrho\sim\cR_\mu^{\bmx}}}[\bmx=x \mid \text{$(v,p)$ is reached on $(\bmx, \bmrho)$}]\\
&= \frac{p(x)}{\sum_{x' \in \ZO^n} p(x')} \text{ where } p(x) \coloneqq \Pr_{\bmx, \bmrho}[\bmx = x] \cdot \Pr_{\bmx, \bmrho}[\text{$(v,p)$ is reached on $(\bmx, \bmrho)$} \mid \bmx = x].
\end{align*}
To analyse $p(x)$, we have:
\[
\Pr\nolimits_{\substack{\bmx\sim \mu\\ \bmrho\sim\cR_\mu^{\bmx}}}[\bmx = x] = \Pr\nolimits_{\bmx \sim \mu}[\bmx = x] = \prod\nolimits_{j \in [n]} \Pr_{\bmx \sim \mu}[\bmx_j = x_j] = \prod\nolimits_{j \in [n]} 1 - (\delta_j/2) - x_j \cdot (1 - \delta_j)
\]
On the other hand, the second component of $p(x)$ is clearly zero if $M(x, v, p) = 0$. For instance, $v$ cannot be reached if $x$ does not satisfy all equations on the path to $v$. Thus, we have:
\begin{align*}
\Pr\nolimits_{\substack{\bmx\sim \mu\\ \bmrho\sim\cR_\mu^{\bmx}}}[\text{$(v,p)$ is reached on $(\bmx, \bmrho)$} \mid \bmx = x] &= \Pr\nolimits_{\bmrho\sim\cR_\mu^{x}}[\text{$(v,p)$ is reached on $(x, \bmrho)$}]\\
&= M(x, v, p) \cdot \Pr\nolimits_{\bmrho\sim\cR_\mu^{x}}[\bmrho \sim p]\\
&= M(x, v, p) \cdot \prod\nolimits_{j \in S^p_0} \frac{2-2\delta_j}{2-\delta_j} \cdot \prod\nolimits_{j\in S^p_\star} \left(\frac{\delta_j}{2-\delta_j}\right)^{1-x_j}.
\end{align*}
Combining those two observations, we get:
\[
p(x) = M(v,p,x) \cdot \prod\nolimits_{j \in S^p_?} \big(1 - \delta_j/2 - x_j \cdot (1 - \delta_j)\big) \cdot \prod\nolimits_{j \in S^p_0} (1 - \delta_j) \cdot \prod\nolimits_{j \in S^p_\star} \delta_j/2.
\]
Observe that the last two products do not involve $x$ at all and can thus be cancelled in the initial expression:
\begin{align*}
\Pr_{\bmx \sim \cX^{v,p}}[\bmx = x] &= \frac{p'(x)}{\sum_{x'} p'(x)} \text{ where } p'(x) = M(x, v, p) \cdot \prod\nolimits_{j \in S^p_?} \big(1 -\delta_j/2 - x_j \cdot (1 - \delta_j)\big). 
\end{align*}
Finally, observe that $M(x, v, p)$ fixes the value of all the bits of $x$ except for $S^p_?$. Thus, the summation in the denominator equals 1 and the claim follows.
\end{proof}
\bigskip
\subsubsection*{Acknowledgements}
We thank Farzan Byramji for useful comments on an earlier version of this paper. M.G., G.M., and W.Y.\ are supported by the Swiss State Secretariat for Education, Research and Innovation (SERI) under contract number MB22.00026.  T.B., and S.G.\ are supported by the National Natural Science Foundation of China Grant No.62102260, NYTP Grant No.20121201, and NYU Shanghai Boost Fund.
\bigskip
\DeclareUrlCommand{\Doi}{\urlstyle{sf}}
\renewcommand{\path}[1]{\small\Doi{#1}}
\renewcommand{\url}[1]{\href{#1}{\small\Doi{#1}}}
\bibliographystyle{alphaurl}
\bibliography{refs}
\end{document}